\documentclass[11pt,letterpaper]{article}

\usepackage[margin=1in]{geometry}

\usepackage{amsmath}
\usepackage{amssymb}
\usepackage{mathtools}
\usepackage{amsthm}

\newtheorem{proposition}{Proposition}

\newtheorem{corollary}{Corollary}
\theoremstyle{definition}

\usepackage{algorithm}
\usepackage{algorithmicx}
\usepackage{algpseudocode}

\usepackage{graphicx}
\graphicspath{{figs/}{./}}

\usepackage[T1]{fontenc}
\usepackage{lmodern}
\usepackage{microtype}
\usepackage{xcolor}
\usepackage[font=small,labelfont=bf]{caption}

\usepackage{authblk}

\usepackage[numbers,sort&compress]{natbib}
\usepackage[colorlinks=true,
            linkcolor=blue!50!black,
            citecolor=blue!50!black,
            urlcolor=blue!60!black,
            breaklinks=true]{hyperref}
\usepackage{orcidlink}

\title{\bfseries\Large Fast Voxelwise SNR Estimation \\
       for Iterative MRI Reconstructions}

\author[1,2]{Onat Dalmaz\,\orcidlink{0000-0001-7978-5311}}
\author[1,2]{Daniel Abraham\,\orcidlink{0009-0004-4299-3132}}
\author[1,2]{Alexander R. Toews\,\orcidlink{0009-0007-6682-1708}}
\author[2,4]{Akshay S. Chaudhari\,\orcidlink{0000-0002-3667-6796}}
\author[1,2]{Kawin Setsompop\,\orcidlink{0000-0003-0455-7634}}
\author[1,2,3]{Brian A. Hargreaves\,\orcidlink{0000-0003-0982-3508}}

\affil[1]{Department of Electrical Engineering, Stanford University, California, USA}
\affil[2]{Department of Radiology, Stanford University, California, USA}
\affil[3]{Department of Bioengineering, Stanford University, California, USA}
\affil[4]{Department of Biomedical Data Science, Stanford University, California, USA}
\affil[ ]{\vspace{0.4em}\normalfont\textit{Correspondence:} \texttt{onat@stanford.edu}}

\date{\today}

\begin{document}
\maketitle

\begin{abstract}
\noindent\textbf{Purpose.}
To develop a fast, general-purpose framework for voxelwise noise
characterization in MRI reconstructions, applicable to both linear and
nonlinear iterative reconstructions. The framework recovers the voxelwise
image-domain noise variance---the primary quantity from which
signal-to-noise ratio (SNR), $g$-factor, and related image-quality metrics
are derived---addressing both the computational intractability of
closed-form formulas beyond Cartesian sampling and the excessive runtime
of empirical Pseudo Multiple Replica (PMR) methods.

\smallskip
\noindent\textbf{Methods.}
A novel noise estimator---which we term PICO (Probing Image-space
COvariance)---is proposed that operates entirely in the image domain by
probing the image-domain noise covariance operator---or, for nonlinear
compressed-sensing reconstructions, the Jacobian of the converged
solution---with random probe images. Complex random-phase probes are
shown theoretically and empirically to minimize estimator variance
compared with Gaussian or real-valued alternatives. The method was
validated against analytical benchmarks and high-replica PMR references
using retrospective Cartesian knee data ($R=2$), prospective non-Cartesian
spiral brain phantom data ($R=2,3,4$), and compressed-sensing knee
reconstructions ($R=2$).

\smallskip
\noindent\textbf{Results.}
In Cartesian experiments, PICO accurately reproduced analytical SENSE
$g$-factor maps. In non-Cartesian spiral imaging ($R=2$), it achieved
$1\%$ estimation error in $64$~s compared with $462$~s for PMR
($\approx 7.2\times$ speedup), with the efficiency advantage persisting
at higher acceleration factors. For nonlinear compressed sensing, the
Jacobian-based estimator produced noise maps consistent with PMR while
converging faster ($52$~s vs.\ $95$~s; $\approx 1.8\times$ speedup).

\smallskip
\noindent\textbf{Conclusion.}
The proposed stochastic framework provides a computationally efficient
alternative to PMR for voxelwise noise and $g$-factor estimation,
enabling rapid SNR characterization for generalized MRI
reconstructions---including non-Cartesian, regularized, and nonlinear
pipelines---where closed-form noise analysis is otherwise unavailable.
By reusing existing reconstruction primitives, the method enables
voxelwise noise maps to be produced as a routine by-product of the
reconstruction pipeline itself.

\end{abstract}

\medskip
\hrule
\medskip
\section{Introduction}
\label{sec:intro}

Signal-to-noise ratio (SNR) is a fundamental measure 
of image quality in MRI \cite{macovski_noise,roemer1990nmr,noise_influence}. Every 
reconstructed image---whether from a fully sampled single-coil 
acquisition, a multi-coil parallel imaging scan, or a 
heavily accelerated compressed-sensing protocol---carries
noise whose spatial structure and local magnitude determines 
the reliability of each voxel for diagnosis, quantification, 
and downstream analysis \cite{SNR_measure,image_recon_SNR}. Voxelwise noise variance maps 
provide the most complete characterization of this 
variability: global summary metrics such as mean SNR or 
contrast-to-noise ratio can mask significant local noise 
amplification that impedes detection of subtle pathologies, 
even when aggregate quality appears acceptable 
\cite{Rubenstein1997,Lerski1993}. Accurate 
voxelwise noise maps are therefore essential for quality 
assurance, protocol optimization, comparison of 
reconstruction strategies, and principled selection of 
regularization parameters 
\cite{image_recon_SNR,Kiryu2023,Knoll2020, toniq}.


During the foundational era of parallel imaging, noise 
analysis was tightly integrated with reconstruction 
development. Pruessmann et~al.\ derived a closed-form 
geometry factor ($g$-factor) alongside the original SENSE 
formulation \cite{sense}, providing not only an analytical 
voxelwise noise amplification map for uniformly sampled 
Cartesian acquisitions but also the conceptual framework 
through which voxelwise SNR has been quantified in 
parallel imaging ever since. Analogous expressions were 
developed for SMASH \cite{smash_noise} and GRAPPA 
\cite{Breuer2009}, and Kellman and McVeigh 
showed that images could be reconstructed directly in SNR 
units using pre-scan noise calibration \cite{image_recon_SNR}. 
These tools were instrumental in algorithm design and 
clinical validation \cite{sense_clinic,Goerner2011}.  As reconstruction pipelines have grown more complex---moving from 
direct-inversion Cartesian methods to iterative solvers for 
non-Cartesian trajectories 
\cite{Pruessmann2001,generalized_sense,cg_sense_revisited}, 
compressed sensing with non-smooth regularizers 
\cite{cs_mri_review,shreyas_cs}, and deep-learning-based 
pipelines \cite{unrolled_1,Heckel2024}---the encoding 
operators (NUFFT, gridding, variable-density compensation) 
have become difficult to form or invert at scale 
\cite{Pruessmann2001,Pruessmann1999,fessler2005}, and the 
closed-form formulas that once made noise characterization 
routine are simply inapplicable. As a 
result, much of the reconstruction literature over the past 
two decades has either omitted noise analysis entirely or 
relegated it to post hoc characterization, decoupled from reconstruction 
development \cite{simulation,montin2021snr}.

In the absence of closed-form noise expressions for these modern reconstruction pipelines, the principal alternative has been the Pseudo Multiple 
Replica (PMR) method \cite{pmr}, which treats the 
reconstruction as an opaque black box: synthetic noise is 
injected into $k$-space, a full reconstruction is re-run for 
each replica, and voxelwise variance is estimated from 
the image ensemble. For computationally intensive 
reconstructions---such as non-Cartesian iterative methods, 
Tikhonov-regularized CG-SENSE, and compressed sensing 
solvers---this requires a complete iterative solve per 
replica \cite{dalmaz2024,akcakaya2014,simulation}. Beyond the computational cost of 
repeating the full solve for every replica 
\cite{dalmaz2024,akcakaya2014,simulation}, PMR is 
fundamentally limited by statistical inefficiency. Because it 
discards all structural knowledge of the reconstruction 
operator---in particular, known factorizations of the covariance matrix---PMR estimates 
variance from squared magnitudes of complex-Gaussian outputs, 
whose heavy-tailed (sub-exponential) distribution causes 
disproportionately slow convergence in high-noise-amplification 
regions where accurate characterization matters most 
(Appendix~\ref{app:tail_bounds}). The resulting wall-clock times have limited 
PMR adoption in practice and prevented its integration into 
real-time or inline workflows 
\cite{pmr,montin2021snr,daude2024inlineqc}. To reduce computational cost, Wiens et~al.\ proposed a generalized pseudo-replica method that replaces the 
$N$-replica ensemble with a \emph{spatial} average: a 
single noise-injected reconstruction is differenced 
against an unperturbed reference, and voxelwise variance 
is estimated from the standard deviation of the 
difference image within a sliding local 
window~\cite{wiens2011}. This delivers a sizable speedup 
over PMR but introduces three structural limitations. 
First, the estimator is fundamentally approximate: at 
every voxel $k$, it converges not to the true voxelwise 
variance $(\boldsymbol{\Sigma}_{\hat{\mathbf{x}}})_{kk}$ 
but to a local spatial average over the window, so the 
resulting noise map is inherently smoothed and biased 
wherever the g-factor varies within the averaging 
volume---precisely the high-g-factor boundaries where 
accurate characterization matters most. Second, the 
pointwise squared-deviation samples remain 
sub-exponential, so spatial averaging inherits the same 
heavy-tailed per-sample variance that governs PMR 
(Appendix~\ref{app:tail_bounds}); the method trades 
temporal replicas for spatial neighbors without escaping 
the underlying statistical inefficiency. Third, the 
subtraction step relies on a linearity argument that 
does not extend to non-smooth regularizers such as 
total variation or $\ell_1$, confining the method to 
linear reconstructions. The generalized pseudo-replica 
method therefore reduces the replica count at the cost 
of a structurally biased, resolution-limited estimator 
confined to linear pipelines, without exploiting any 
knowledge of the reconstruction operator.

Recent work has begun to address noise estimation in 
deep-learning-based reconstructions by locally linearizing 
the converged reconstruction map and using automatic 
differentiation to form Jacobian--vector products. This 
approach has been applied to unrolled deep networks 
\cite{dalmaz2025efficient} and to $k$-space interpolation 
(RAKI) networks \cite{dawood,homolya2025raki}, 
representing important progress toward tractable noise 
analysis for modern learned reconstructions. These methods, 
however, have three limitations that leave the broader 
problem unresolved. First, they are inherently approximate: 
the Jacobian captures only the first-order behavior of the 
reconstruction map at a single operating point, and the 
quality of the resulting variance estimate depends on how 
faithfully this linearization reflects the true nonlinear 
response. Second, they do not exploit the known structure of 
the encoding operator and its adjoint, treating the 
reconstruction as a generic differentiable map rather than 
as the solution of a specific physics-driven inverse problem. 
Third, and most importantly, they have not been developed 
for iterative linear reconstructions or for regularized 
compressed-sensing reconstructions---despite the fact that 
CG-SENSE and closely related linear iterative methods remain 
the most widely deployed reconstruction class in both 
clinical and research practice, and admit an \emph{exact}, 
fully analytical characterization of the image-domain noise 
covariance. A framework that 
directly exploits this known structure---recovering the 
exact covariance for linear reconstructions while extending 
naturally to nonlinear ones---has so far been absent.

In this work, we invert the standard paradigm of MRI 
noise characterization. Rather than injecting synthetic 
noise into $k$-space and averaging over many full 
reconstructions, we leave the acquired data untouched 
and probe the image-domain noise-propagation operator 
directly with stochastic image-space probes---recovering 
voxelwise variance from how the operator responds, 
rather than from how reconstructions fluctuate. We call 
this framework PICO (Probing Image-space COvariance). 
For linear reconstructions such as CG-SENSE with or 
without regularization---including the unaccelerated 
reference case that underlies essentially all 
quantitative SNR studies---PICO is mathematically exact: 
it recovers the diagonal of the noise covariance 
without any linearization, Jacobian approximation, or 
model simplification. The only source of error is 
finite-sample Monte Carlo variance, which decreases as 
$1/N$ and, unlike PMR's heavy-tailed sub-exponential 
behavior, concentrates with sub-Gaussian tails whose 
decay rate is \emph{independent} of the local noise 
amplitude (Appendix~\ref{app:tail_bounds}). Empirically, 
this translates into a substantial efficiency advantage 
over PMR that is retained across acceleration factors: 
as the reconstruction becomes more ill-conditioned, PMR 
requires disproportionately more additional replicas to 
maintain a given accuracy than PICO requires additional 
probes, so PICO scales more gracefully in the regimes 
where accurate noise characterization is hardest. 
Crucially, PICO requires no new computational 
infrastructure: its implicit covariance--vector 
products reuse the exact CG-SENSE primitives already 
present in any iterative reconstruction 
pipeline---forward and adjoint operators, regularized 
normal-equation solves, NUFFT, Toeplitz embeddings, and 
any physics corrections (off-resonance, density 
compensation) the reconstruction supports.

The same framework extends cleanly to nonlinear 
reconstructions (e.g., compressed sensing with 
total-variation or $\ell_1$ regularization) by applying the 
estimator to a first-order Jacobian linearization of the 
converged reconstruction map via automatic-differentiation 
Jacobian--vector products, unifying the linear exact and 
nonlinear approximate regimes under a single algorithm. 
Crucially, the choice of probe distribution in both regimes 
is not a heuristic. Building on the classical stochastic 
diagonal estimators of Hutchinson and Bekas et~al.\ 
\cite{hutchinson,bekas}, we generalize the variance analysis 
from real symmetric to complex Hermitian positive-semidefinite 
operators and prove that \emph{random-phase 
probes}---unit-magnitude complex vectors with uniformly 
distributed phase---attain the minimum achievable per-sample 
estimator variance among all unit-variance probe 
distributions, strictly outperforming both standard complex 
Gaussian and real-valued Rademacher alternatives 
(Appendix~\ref{app:variance-derivation}). To our knowledge, 
this kurtosis-optimal probe structure has not previously been 
exploited for MRI noise characterization. Taken together, 
these contributions re-establish voxelwise SNR 
characterization as a lightweight, physics-grounded 
by-product of the reconstruction pipeline itself---available across a broad range of reconstruction types including direct-inversion Cartesian SENSE, iterative CG-SENSE on arbitrary trajectories, Tikhonov-regularized linear solvers, and nonlinear compressed-sensing reconstructions

To demonstrate PICO, we apply it to three representative MRI reconstruction settings. On retrospectively undersampled Cartesian knee data at $R=2$, 
the estimator is compared against the closed-form analytical 
SENSE  noise reference, providing a ground-truth benchmark where 
one exists. On prospective non-Cartesian spiral brain phantom 
data at $R \in \{2,3,4\}$, where no closed-form reference is 
available, we use a high-replica PMR estimate whose convergence is 
independently certified (Appendix~\ref{app:pmr_convergence}); 
benchmarking against PMR's own converged output is 
deliberately conservative, since any residual PMR-specific 
bias is absorbed into the reference rather than penalized. 
On compressed-sensing reconstructions of the fastMRI knee 
dataset at $R=2$, a total-variation--regularized FISTA 
pipeline exercises the Jacobian extension in a representative 
nonlinear setting. An ablation study empirically confirms the 
kurtosis-based probe-optimality theory, and a robustness 
study demonstrates that the Jacobian approximation remains 
accurate across the full range of SNRs encountered in 
standard acquisitions, with visible departures from PMR 
appearing only at stress-test noise levels where the 
underlying reconstruction itself becomes strongly nonlinear.

\section{Theory}
\label{sec:theory}
\subsection{MRI Reconstruction and Noise}

\subsubsection{Acquisition Model}
Let $\mathbf{x}\in\mathbb{C}^{N_x}$ be the image and $\mathbf{b}\in\mathbb{C}^{N_k}$ the stacked multi–coil k–space data. The generalized linear encoding operator $\widetilde{\mathbf{A}}\in\mathbb{C}^{N_k\times N_x}$ includes coil sensitivities, Fourier/NUFFT, density compensation, and sampling trajectory:
\begin{equation}
\mathbf{b}=\widetilde{\mathbf{A}}\mathbf{x}+\mathbf{n}.
\label{eq:forward_model_original_theory}
\end{equation}
\subsubsection{Multi-Coil Noise and Pre-Whitening}
Acquisition noise is modeled as circular complex Gaussian with multi-channel noise covariance matrix
$\boldsymbol{\Sigma}_{\mathbf{n}}$:
\begin{equation}
\mathbf{n}\sim\mathcal{CN}\!\left(\mathbf{0},\,\boldsymbol{\Sigma}_{\mathbf{n}}\right),\qquad \boldsymbol{\Sigma}_{\mathbf{n}}\succeq \mathbf{0}.
\label{eq:input_noise_theory}
\end{equation}
Let $\mathbf{L}\mathbf{L}^{\mathrm{H}}=\boldsymbol{\Sigma}_{\mathbf{n}}$ (e.g., Cholesky factorization). Left-multiplying \eqref{eq:forward_model_original_theory} by $\mathbf{L}^{-1}$ yields the pre-whitened system
\begin{equation}
\mathbf{b}_w=\mathbf{A}\mathbf{x}+\mathbf{n}_w,\quad \mathbf{A}=\mathbf{L}^{-1}\widetilde{\mathbf{A}},\quad \mathbf{n}_w=\mathbf{L}^{-1}\mathbf{n}\sim\mathcal{CN}(\mathbf{0},\mathbf{I}).
\label{eq:prewhitened_system_theory}
\end{equation}
All subsequent derivations and experiments operate on the 
pre-whitened system of Eq.~\eqref{eq:prewhitened_system_theory}. 
Pre-whitening is a standard first step in multi-coil reconstruction pipelines~\cite{hansen2015image,pmr}: it folds the inter-coil noise correlations into the encoding operator through $\mathbf{L}^{-1}$, leaving residual noise that is spatially white and unit-variance across channels. As a result, the image-domain noise covariance of any linear reconstruction depends only on the reconstruction operator $\mathbf{R}$ acting on the pre-whitened data (Eq.~\eqref{eq:cov_general_theory}), and the proposed estimator probes $\boldsymbol{\Sigma}_{\hat{\mathbf{x}}} = \mathbf{R}\mathbf{R}^{\mathrm{H}}$ without requiring explicit knowledge of the original channel noise correlations.

\subsubsection{Reconstruction model and operators: from least squares to CG-SENSE and beyond}

In practice, pre-whitened multi-coil MRI data are modeled as \(\mathbf{b}_w=\mathbf{A}\mathbf{x}+\mathbf{n}_w\) with \(\mathbf{n}_w\sim\mathcal{CN}(\mathbf{0},\mathbf{I})\). A standard workhorse in clinical and research pipelines is a quadratic (least-squares) objective with optional \(\ell_2\) (Tikhonov) stabilization:
\begin{equation}
\hat{\mathbf{x}}_\lambda \in \arg\min_{\mathbf{x}}\; \tfrac{1}{2}\|\mathbf{A}\mathbf{x}-\mathbf{b}_w\|_2^2 + \tfrac{\lambda}{2}\|\mathbf{x}\|_2^2, \qquad \lambda \ge 0.
\label{eq:tik_objective}
\end{equation}
This objective is convex and \emph{quadratic}; its minimizer depends \emph{linearly} on the data. Equivalently, there exists a linear reconstruction operator \(\mathbf{R}_\lambda\) such that
\begin{equation}
\hat{\mathbf{x}}_\lambda = \mathbf{R}_\lambda\,\mathbf{b}_w,
\qquad
\mathbf{R}_\lambda = (\mathbf{A}^{\mathrm{H}}\mathbf{A}+\lambda\mathbf{I})^{-1}\mathbf{A}^{\mathrm{H}}.
\label{eq:R_lambda}
\end{equation}
Forming \((\mathbf{A}^{\mathrm{H}}\mathbf{A}+\lambda\mathbf{I})^{-1}\) explicitly is unnecessary in large-scale MRI \cite{yeung2025algebraic}. Instead, one solves the normal equations
\begin{equation}
(\mathbf{A}^{\mathrm{H}}\mathbf{A}+\lambda\mathbf{I})\,\mathbf{x}=\mathbf{A}^{\mathrm{H}}\mathbf{b}_w,
\label{eq:normal_equations}
\end{equation}
with conjugate gradients (CG). This is precisely the CG-SENSE procedure widely used in routine workflows: it exploits fast applications of \(\mathbf{A}\) and \(\mathbf{A}^{\mathrm{H}}\) (FFT/NUFFT + sensitivity maps) and avoids matrix factorizations \cite{cg_sense_revisited,Pruessmann2001}. 

\subsubsection{Noise propagation and covariance}

With pre-whitened noise \(\mathbf{n}_w\sim\mathcal{CN}(\mathbf{0},\mathbf{I})\) and the linear reconstruction \(\hat{\mathbf{x}}_\lambda=\mathbf{R}_\lambda\,\mathbf{b}_w\) from Eq.~\eqref{eq:R_lambda}, the image-domain noise \emph{covariance} matrix (not correlation matrix---the diagonal entries give voxelwise noise variances, not unity) follows directly as
\begin{equation}
\boldsymbol{\Sigma}_{\hat{\mathbf{x}}}
=\mathrm{Cov}[\hat{\mathbf{x}}_\lambda]
=\mathbf{R}_\lambda\,\mathbf{R}_\lambda^{\mathrm{H}}.
\label{eq:cov_general_theory}
\end{equation}
Substituting \(\mathbf{R}_\lambda\) yields the familiar expressions
\begin{equation}
\boldsymbol{\Sigma}_{\hat{\mathbf{x}}}
=(\mathbf{A}^{\mathrm{H}}\mathbf{A}+\lambda\mathbf{I})^{-1}\,\mathbf{A}^{\mathrm{H}}\mathbf{A}\,(\mathbf{A}^{\mathrm{H}}\mathbf{A}+\lambda\mathbf{I})^{-1},
\qquad \lambda\ge 0,
\label{eq:cov_tikh_theory}
\end{equation}
and, in the unregularized case \((\lambda=0)\) with invertible \(\mathbf{A}^{\mathrm{H}}\mathbf{A}\),
\begin{equation}
\boldsymbol{\Sigma}_{\hat{\mathbf{x}}}=(\mathbf{A}^{\mathrm{H}}\mathbf{A})^{-1}.
\label{eq:cov_ls_theory}
\end{equation}
The voxelwise noise variance map is the diagonal \cite{sense,Breuer2009}
\begin{equation}
\boldsymbol{\sigma}^2_{\hat{\mathbf{x}}}=\mathrm{diag}\!\big(\boldsymbol{\Sigma}_{\hat{\mathbf{x}}}\big).
\label{eq:varmap_theory}
\end{equation}
Relative noise amplification (the \(g\)-factor) is then defined as:
\begin{equation}
g(i)=\frac{\sigma_{\hat{\mathbf{x}},\mathrm{acc},\lambda}(i)}
  {\sqrt{R}\;\sigma_{\hat{\mathbf{x}},\mathrm{ref}}(i)},
\label{eq:g_factor_def}
\end{equation}
where $\sigma_{\hat{\mathbf{x}},\mathrm{acc},\lambda}(i)$ is the 
noise standard deviation of the $i$-th voxel in the accelerated, 
regularized reconstruction, and 
$\sigma_{\hat{\mathbf{x}},\mathrm{ref}}(i)$ is the corresponding 
quantity from a fully sampled, reference reconstruction using the same coil configuration.


\subsection{Direct Noise Estimation via Probing the Image-Space Covariance (PICO)}
\label{subsec:our_method}
\begin{figure*}[t]
\centering
\includegraphics[width=\textwidth]{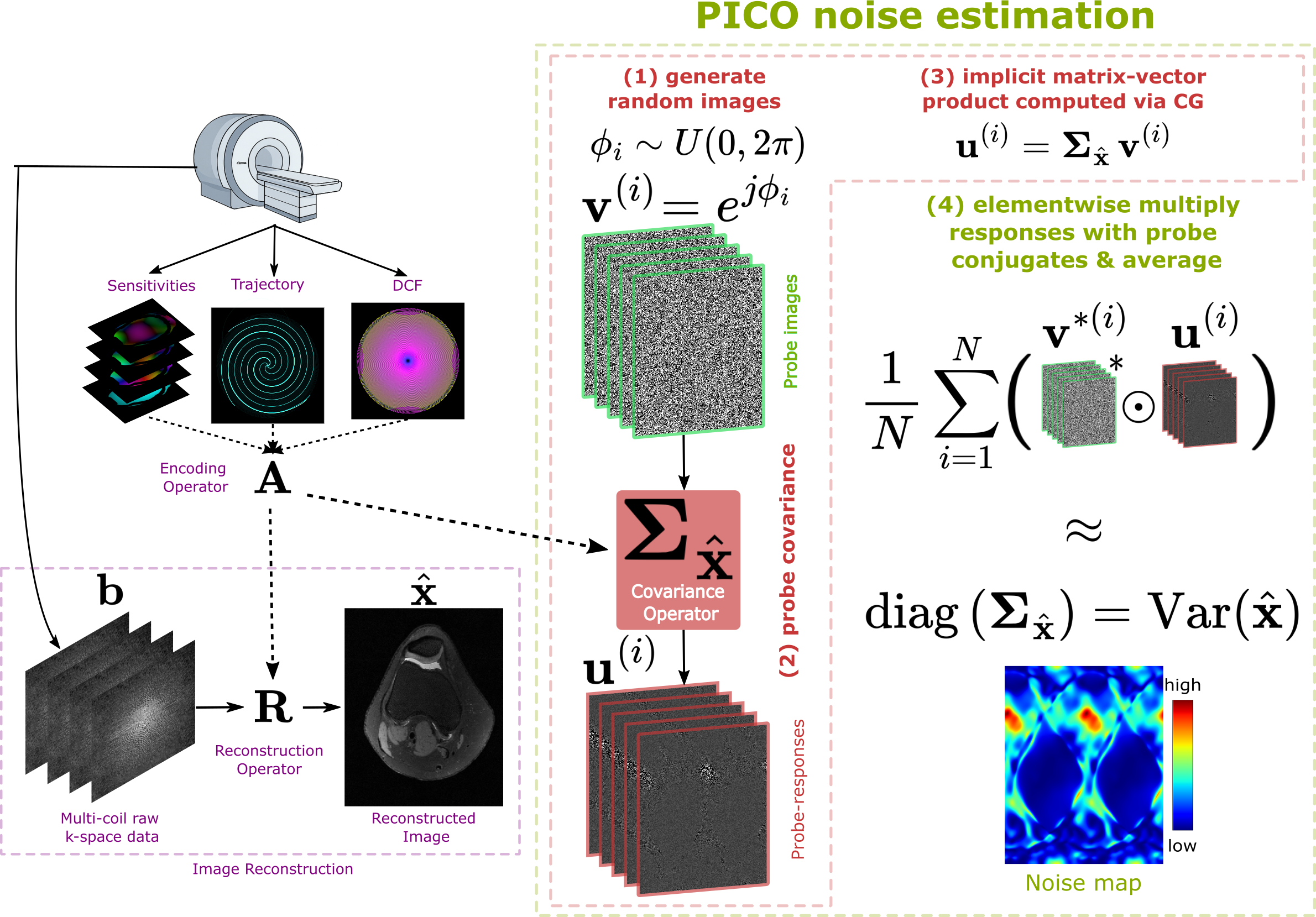}
\caption{Overview of PICO for fast voxelwise noise 
characterization. Left: the MRI acquisition is modeled 
by an encoding operator $\mathbf{A}$ (coil sensitivities, 
sampling trajectory, density compensation); the 
reconstruction operator $\mathbf{R}$ maps multi-coil 
$k$-space data to the image estimate~$\hat{\mathbf{x}}$. 
Right: PICO estimates the voxelwise noise variance map 
$\mathrm{diag}(\boldsymbol{\Sigma}_{\hat{\mathbf{x}}})$ 
in four steps. 
(1)~Generate unit-magnitude random-phase probe images 
$\mathbf{v}^{(i)}$, $i=1,\dots,N$. 
(2)~Pass each probe through the implicit covariance 
operator 
$\boldsymbol{\Sigma}_{\hat{\mathbf{x}}} 
= \mathbf{R}\mathbf{R}^{\mathrm{H}}$, 
without forming the matrix explicitly. 
(3)~Evaluate the covariance--vector product 
$\mathbf{u}^{(i)} 
= \boldsymbol{\Sigma}_{\hat{\mathbf{x}}}\,\mathbf{v}^{(i)}$ 
via CG on the normal equations, reusing the same 
$\mathbf{A}$/$\mathbf{A}^{\mathrm{H}}$ primitives as 
CG-SENSE reconstruction. 
(4)~Elementwise multiply each response by its probe 
conjugate and average: 
$\hat{\boldsymbol{\sigma}}^{2}_{\hat{\mathbf{x}}} 
= \frac{1}{N}\sum_{i=1}^{N} 
\mathbf{v}^{(i)*}\!\odot\mathbf{u}^{(i)}$.}
\label{fig:noise_module}
\end{figure*}

To overcome the limitations of PMR, we propose a direct and computationally efficient framework 
for estimating the voxelwise noise variance map 
$\boldsymbol{\sigma}^2_{\hat{\mathbf{x}}}=\mathrm{diag}(\boldsymbol{\Sigma}_{\hat{\mathbf{x}}})$ 
defined in \eqref{eq:varmap_theory}, without performing multiple reconstructions. 
PICO operates entirely in the image domain and is summarized in Fig.~\ref{fig:noise_module}.  
It builds upon stochastic trace and diagonal estimators from numerical linear algebra 
\cite{hutchinson, bekas}, generalized here to complex Hermitian covariances with
$\boldsymbol{\Sigma}_{\hat{\mathbf{x}}}=\mathbf{R}\mathbf{R}^{\mathrm{H}}$.

\subsubsection{An unbiased noise variance estimator}
\label{sec:estimator}
The goal is to estimate $\mathrm{diag}(\boldsymbol{\Sigma}_{\hat{\mathbf{x}}})$
without explicitly forming the full covariance matrix. 
Let $\{\mathbf{v}^{(i)}\}_{i=1}^{N_{\text{probes}}}$ be independent probe images drawn from a distribution
with zero mean and identity covariance (examples below). 
For each probe, define the implicit covariance–vector product
\begin{equation}
\mathbf{u}^{(i)} \;\coloneqq\; \boldsymbol{\Sigma}_{\hat{\mathbf{x}}}\,\mathbf{v}^{(i)},
\label{eq:def-u}
\end{equation}
which can be evaluated using the same primitives as CG-SENSE (applications of $\mathbf{A}$, $\mathbf{A}^{\mathrm{H}}$, and regularized normal-equation solves) without forming $\boldsymbol{\Sigma}_{\hat{\mathbf{x}}}$.

The diagonal of $\boldsymbol{\Sigma}_{\hat{\mathbf{x}}}$ is then estimated via elementwise products of probe conjugates and responses:
\begin{equation}
\widehat{\mathrm{diag}}\!\left(\boldsymbol{\Sigma}_{\hat{\mathbf{x}}}\right)
\;\;=\;\;
\frac{1}{N_{\text{probes}}}\sum_{i=1}^{N_{\text{probes}}}
\mathbf{v}^{(i)*}\odot \mathbf{u}^{(i)}.
\label{eq:diagonal_estimator}
\end{equation}
By identification with \eqref{eq:varmap_theory}, we write the \emph{estimated noise variance map} as
\begin{equation}
\hat{\boldsymbol{\sigma}}^{2}_{\hat{\mathbf{x}}}
\;\coloneqq\;
\widehat{\mathrm{diag}}\!\left(\boldsymbol{\Sigma}_{\hat{\mathbf{x}}}\right)
\;\;=\;\;
\frac{1}{N_{\text{probes}}}\sum_{i=1}^{N_{\text{probes}}}
\mathbf{v}^{(i)*}\odot \Bigl(\boldsymbol{\Sigma}_{\hat{\mathbf{x}}}\,\mathbf{v}^{(i)}\Bigr).
\label{eq:varmap_estimator}
\end{equation}

\subsubsection{Unbiasedness}
Assume the probe distribution satisfies
\begin{equation}
\mathbb{E}\!\left[\mathbf{v}^{(i)}\right]=\mathbf{0},
\qquad
\mathbb{E}\!\left[\mathbf{v}^{(i)}(\mathbf{v}^{(i)})^{\mathrm{H}}\right]=\mathbf{I}.
\label{eq:probe_conditions}
\end{equation}
Let $\mathbf{y}^{(i)}=\mathbf{v}^{(i)*}\odot \mathbf{u}^{(i)}$ with $\mathbf{u}^{(i)}=\boldsymbol{\Sigma}_{\hat{\mathbf{x}}}\mathbf{v}^{(i)}$.
For the $k$-th voxel,
\[
y^{(i)}_k = v_k^{(i)*}\!\sum_{j}(\boldsymbol{\Sigma}_{\hat{\mathbf{x}}})_{kj} v^{(i)}_j
\]
\begin{equation}
\mathbb{E}\!\left[y^{(i)}_k\right]
= \sum_{j}(\boldsymbol{\Sigma}_{\hat{\mathbf{x}}})_{kj}\,\mathbb{E}\!\left[v_k^{(i)*} v^{(i)}_j\right]
= (\boldsymbol{\Sigma}_{\hat{\mathbf{x}}})_{kk}.
\label{eq:unbiasedness}
\end{equation}

Thus $\mathbb{E}\!\left[\mathbf{y}^{(i)}\right]=\mathrm{diag}(\boldsymbol{\Sigma}_{\hat{\mathbf{x}}})$, and averaging
in \eqref{eq:varmap_estimator} yields an unbiased estimator of the noise variance map:
\begin{equation}
\mathbb{E}\!\left[\hat{\boldsymbol{\sigma}}^{2}_{\hat{\mathbf{x}}}\right]
\;=\;
\mathrm{diag}\!\left(\boldsymbol{\Sigma}_{\hat{\mathbf{x}}}\right)
\;=\;
\boldsymbol{\sigma}^{2}_{\hat{\mathbf{x}}}.
\label{eq:unbiasedness_result}
\end{equation}
The estimator variance decreases as $1/N_{\text{probes}}$. 
For a detailed analysis of the theoretical properties of the estimator, 
see Appendix~\ref{app:variance-derivation} and \ref{app:tail_bounds}.

\subsubsection{Implicit matrix–vector product via CG}
The vector $\mathbf{u}^{(i)}$ in \eqref{eq:def-u} is not computed by explicitly forming 
$\mathbf{\Sigma}_{\hat{\mathbf{x}}}$, which would be intractable for realistic image sizes. 
Instead, it is obtained implicitly using operator access to the reconstruction operator 
$\mathbf{R}$ and its adjoint $\mathbf{R}^\mathrm{H}$, both of which are naturally 
available within iterative solvers such as CG-SENSE. 
Since $\mathbf{R}$ itself corresponds to a linear mapping, 
the product $\mathbf{u}^{(i)} = \mathbf{\Sigma}_{\hat{\mathbf{x}}}\mathbf{v}^{(i)} 
= \mathbf{R}(\mathbf{R}^\mathrm{H}\mathbf{v}^{(i)})$ can be efficiently evaluated 
through existing forward and adjoint operations without ever forming $\mathbf{\Sigma}_{\hat{\mathbf{x}}}$ 
explicitly. In practice, this implicit computation is implemented using the same low-level routines that perform the CG-SENSE reconstruction, requiring only matrix–vector products with 
$\mathbf{A}$ and $\mathbf{A}^\mathrm{H}$ and regularized normal-equation solves. A complete algoritmhic pseudocode description of PICO is provided in Appendix~\ref{app:PICO_algorithm}.

\subsubsection{Choice of probing vectors}
\label{sec:probes}
Several families of random vectors satisfy the probe conditions in \eqref{eq:probe_conditions}. The single-sample variance of the diagonal estimator at each voxel decomposes into two contributions (Appendix~\ref{app:variance-derivation}): one that scales with the kurtosis $\kappa = \mathbb{E}[|v_k|^4]$ of the probe distribution and multiplies the squared diagonal entries $|(\boldsymbol{\Sigma}_{\hat{\mathbf{x}}})_{kk}|^2$, and one that depends only on the off-diagonal energy $\sum_{\ell\neq k}|(\boldsymbol{\Sigma}_{\hat{\mathbf{x}}})_{k\ell}|^2$ and is invariant across probe distributions. Minimizing the per-sample variance therefore reduces to minimizing the kurtosis $\kappa$ subject to the unit-variance constraint $\mathbb{E}[|v_k|^2] = 1$.

A natural baseline is the standard complex Gaussian probe,
\begin{equation}
\mathbf{v}^{(i)} \sim \mathcal{CN}(\mathbf{0}, \mathbf{I}),
\label{eq:gaussian_probe}
\end{equation}
which satisfies $\kappa = 2$. Prior work on real-valued diagonal estimation showed that Rademacher probes~\cite{hutchinson,bekas},
\begin{equation}
v^{(i)}_k \in \{-1, +1\}, \quad \mathbb{P}(v^{(i)}_k = +1) = \mathbb{P}(v^{(i)}_k = -1) = \tfrac{1}{2},
\label{eq:real_rademacher}
\end{equation}
attain $\kappa = 1$, halving the diagonal contribution to the per-sample variance relative to Gaussian probes. Since $\kappa = 1$ is the smallest kurtosis achievable by any unit-variance distribution, real Rademacher probes are already kurtosis-optimal for real-valued covariance operators.

For the complex Hermitian covariances $\boldsymbol{\Sigma}_{\hat{\mathbf{x}}}$ that arise in MRI reconstruction, however, real-valued probes are information-limited: a real probe encodes only one real degree of freedom per voxel, leaving the imaginary components of the off-diagonal entries $(\boldsymbol{\Sigma}_{\hat{\mathbf{x}}})_{k\ell}$ invisible to a single sample, so that two independent real probes are required to capture the same information as a single complex probe (Appendix~\ref{app:variance-derivation}). We therefore adopt the complex analogue of the Rademacher distribution: unit-magnitude probes with uniformly distributed phase,
\begin{equation}
v_k^{(i)} = e^{\mathrm{j}\theta_k^{(i)}}, \qquad \theta_k^{(i)} \sim \mathrm{Uniform}[0, 2\pi].
\label{eq:complex_rademacher}
\end{equation}
The random-phase distribution satisfies $\mathbb{E}[\mathbf{v}^{(i)}] = \mathbf{0}$ and $\mathbb{E}[\mathbf{v}^{(i)}(\mathbf{v}^{(i)})^{\mathrm{H}}] = \mathbf{I}$ exactly---these conditions are inherent to the construction rather than asymptotic---and attains the same minimum kurtosis $\kappa = 1$ as real Rademacher while encoding two real degrees of freedom per voxel. Random-phase probes therefore strictly minimize the single-sample variance among unit-variance distributions on complex Hermitian operators (Appendix~\ref{app:variance-derivation}). We empirically validate this design choice in an ablation study (Sec.~\ref{subsec:noncart_cgsense_noise_amp}) comparing random-phase, real Rademacher, and complex Gaussian probes at matched probe counts.

\subsubsection{Extension to Compressed-Sensing reconstructions}
\label{subsec:cs_extension}
When the reconstruction problem in Eq.~\eqref{eq:tik_objective} is 
non-smooth (e.g., TV or $\ell_1$), the reconstruction map 
$f:\mathbf{b}_w\mapsto\hat{\mathbf{x}}$ becomes nonlinear \cite{Lustig2008} and the exact covariance 
$\mathrm{Cov}[\hat{\mathbf{x}}]$ has no tractable closed form. We employ a local linearization, following the Jacobian-based 
framework introduced in recent work 
\cite{dalmaz2025efficient,dawood}. Let 
$\mathbf{k}_0 = \mathbf{A}\mathbf{x}$ denote the noise-free 
k-space data. If the converged solution $\hat{\mathbf{x}}$ is 
differentiable with respect to $\mathbf{b}_w$ at $\mathbf{k}_0$ 
(which holds almost everywhere for piecewise-linear proximal 
operators such as soft-thresholding 
\cite{dalmaz2025efficient,dawood}), then
\[
f(\mathbf{k}_0 + \mathbf{n}_w)\;\approx\; f(\mathbf{k}_0)
  \;+\; \mathbf{J}_f(\mathbf{k}_0)\,\mathbf{n}_w,
\]
and the image-domain noise covariance is approximated as
\[
\boldsymbol{\Sigma}_{\hat{\mathbf{x}}} \;\approx\; 
  \mathbf{J}_f(\mathbf{k}_0)\,
  \mathbf{J}_f(\mathbf{k}_0)^{\mathrm{H}}.
\]
This approximation is exact to first order in the noise amplitude 
and becomes increasingly accurate as the input SNR increases. 
Consequently, the stochastic diagonal estimator applied to 
$\mathbf{J}_f\mathbf{J}_f^{\mathrm{H}}$ yields a 
\emph{locally unbiased} estimate of the noise variance map. Modern automatic differentiation 
frameworks provide efficient access to Jacobian--vector products, 
enabling direct evaluation of 
$\mathbf{u} = \mathbf{J}_f(\mathbf{k}_0)\,\mathbf{v}$ for random 
probes $\mathbf{v}$ \cite{dalmaz2025efficient,dawood,homolya2025raki}.

Therefore, even in the nonlinear CS setting, the stochastic diagonal 
estimator remains directly applicable: averaging 
$|\mathbf{u}|^2$ across probes yields unbiased noise variance 
estimates, and taking the elementwise square root produces voxelwise 
noise standard deviation maps.

\paragraph{Relationship to the linear case.}
The Jacobian formulation above subsumes the linear covariance-probing construction of Sec.~\ref{sec:estimator} as a special case. For a linear reconstruction $\hat{\mathbf{x}} = \mathbf{R}\mathbf{b}_w$, the Jacobian is constant and equal to the reconstruction operator itself, $\mathbf{J}_f(\mathbf{k}_0) = \mathbf{R}$, so that
\begin{equation}
\mathbf{J}_f\mathbf{J}_f^{\mathrm{H}} = \mathbf{R}\mathbf{R}^{\mathrm{H}} = \boldsymbol{\Sigma}_{\hat{\mathbf{x}}},
\label{eq:linear_equivalence}
\end{equation}
and the Jacobian-based estimator coincides exactly with direct probing of $\boldsymbol{\Sigma}_{\hat{\mathbf{x}}}$. The proposed framework is therefore a single estimator applied to the same class of Hermitian positive-semidefinite operators in both regimes---exact for linear reconstructions, first-order for nonlinear ones---rather than two distinct methods for two distinct problems. In practice, we use the direct covariance formulation $\boldsymbol{\Sigma}_{\hat{\mathbf{x}}}\mathbf{v} = \mathbf{R}\mathbf{R}^{\mathrm{H}}\mathbf{v}$ for linear reconstructions because it reuses existing CG-SENSE primitives without invoking automatic differentiation, but the two constructions yield mathematically identical results whenever the reconstruction map is linear.

\section{Methods}


\subsection{Datasets}
\label{sec:methods}

\paragraph{Stanford knee.}
Eight-channel proton-density-weighted 3D fast--spin-echo knee scans were obtained from \texttt{mridata.org} \cite{epperson2013creation}: 14 subjects, acquisition matrix $k_x \times k_y \times k_z = 320 \times 320 \times 256$. All datasets were fully sampled at acquisition. A 1D inverse Fourier transform along readout produced hybrid k-space slices ($x \times k_y \times k_z$). Coil sensitivities were estimated with JSENSE (kernel width 8, $20 \times 20$ auto-calibration region) \cite{jsense}. Quantitative metrics were computed per slice and averaged within each subject; inter-subject statistics are reported over all 14 subjects.

\paragraph{Physical brain phantom.}
A physical brain phantom was imaged using a gradient-echo spiral acquisition at 1\,mm isotropic resolution (22\,cm FOV, 60 interleaves) on a GE 3T Ultra High Performance scanner with 13 receive channels. Coil sensitivities were estimated by ESPIRiT from the fully sampled central region \cite{Uecker2014}.

\paragraph{fastMRI knee.}
Fifteen fully sampled subjects from the fastMRI dataset \cite{fastmri} were used: Cartesian fast spin-echo, 15-channel coronal proton-density fat-suppressed images, matrix $320 \times 320 \times 25$. Coil sensitivities were estimated by ESPIRiT using 24 central auto-calibration lines \cite{Uecker2014}.

\subsection{Reconstruction experiments}

All experiments assess the accuracy of voxelwise noise variance estimation, with $g$-factor maps reported as derived ratios where a fully sampled reference is available. For the Cartesian experiment, a closed-form analytical noise reference serves as ground truth; for the non-Cartesian and CS experiments, where no closed-form reference exists, high-replica PMR maps serve as surrogate references with convergence independently validated (Appendix~\ref{app:pmr_convergence}). Reconstructed magnitude images are shown alongside noise and $g$-factor maps for anatomical context. For brevity, $N$ denotes the number of probes (PICO) or replicas (PMR) throughout.

\subsubsection{Cartesian CG-SENSE}

\paragraph{Experimental setup.}
The Stanford knee dataset was used to evaluate both methods against an analytical ground truth. Fully sampled multi-coil k-space was retrospectively undersampled along the phase-encoding direction at $R=2$ (uniform pattern, $R_x=1$, $R_y=2$). Undersampled data were reconstructed by CG-SENSE without regularization:
\[
\hat{\mathbf{x}} = \arg\min_{\mathbf{x}} \,\|\mathbf{A}_{\mathrm{acc}} \mathbf{x} - \mathbf{b}_{\mathrm{acc}}\|_2^2,
\]
with eigenvalue normalization (15-step power method, 1.01$\times$ safety margin) and 25 CG iterations. The closed-form analytical SENSE noise-variance map served as ground-truth reference \cite{sense}.

\paragraph{Evaluation.}
Both PICO and PMR were evaluated at probe/replica counts $N \in \{10, 20, 40, 80, 200\}$. Quantitative accuracy was assessed via voxelwise NRMSE at $N=200$, computed per slice and averaged within each subject. Statistical significance of the NRMSE difference between PICO and PMR was tested using a paired Wilcoxon signed-rank test across the 14 subjects' per-subject mean NRMSE values (two-sided $p < 0.01$).

\subsubsection{Non-Cartesian CG-SENSE}

\paragraph{Experimental setup.}
non-Cartesian multi-coil k-space data from brain phantom was undersampeld via retrospective shot subsampling: every $R$-th interleave was retained, yielding acceleration factors $R \in \{2, 3, 4\}$. Reconstruction solved the Tikhonov-regularized problem \cite{Pruessmann2001}:
\[
\hat{\mathbf{x}}=\arg\min_{\mathbf{x}} \|\mathbf{A}_{\mathrm{acc}}\mathbf{x}-\mathbf{b}_{\mathrm{acc}}\|_2^2+\lambda \|\mathbf{x}\|_2^2,
\]
with $\lambda=0.1$, Kaiser--Bessel NUFFT (width 4), Toeplitz normal-operator approximation, eigenvalue normalization (15-step power method, 1.01$\times$), and 100 CG iterations (tolerance $10^{-2}$).

\paragraph{Estimators and reference.}
PMR used i.i.d.\ complex Gaussian noise ($\sigma_k = 10^{-2}$); PICO used random-phase probes. Because no closed-form $g$-factor exists for non-Cartesian trajectories, a high-fidelity PMR estimate at $N_{\mathrm{ref}} = 30{,}000$ replicas served as a reference, with convergence independently validated (Appendix~\ref{app:pmr_convergence}). Benchmarking against PMR's own converged output, if anything, favors PMR: any residual PMR-specific bias is absorbed into the reference rather than penalized.

\paragraph{Evaluation.}
Both methods were evaluated from $N=10$ to $N=10,000$ in steps of 10. NRMSE was computed against the PMR reference at each $N$ to generate accuracy--runtime trade-off curves. We report the operating point where each method first reached $1\%$ NRMSE. An ablation study compared convergence of three probe distributions (random phase, real Rademacher, standard complex Gaussian) over the same $N$ range.

\subsubsection{Compressed-Sensing reconstruction}

\paragraph{Experimental setup.}
Ten subjects' images from the fastMRI knee dataset were retrospectively undersampled using a Poisson-disc pattern ($R=2$, $24 \times 24$ calibration region). Reconstruction used FISTA \cite{Beck2009} with total-variation (TV) regularization \cite{knoll_tv}:
\[
\hat{\mathbf{x}} = \arg\min_{\mathbf{x}} \tfrac{1}{2}\|\mathbf{A} \mathbf{x} - \mathbf{b}\|_2^2 + \lambda_{\mathrm{TV}} \|\nabla \mathbf{x}\|_{1},
\]
with $\lambda_{\mathrm{TV}} = 10^{-2}$ and 100 FISTA iterations.

\paragraph{Variance estimation.}
PICO was applied to the Jacobian of the reconstruction map at the converged solution, evaluating $\mathbf{u} = \mathbf{J}_f(\mathbf{k}_0)\mathbf{v}$ via automatic differentiation~\cite{dalmaz2025efficient}. PMR used $\sigma_k = 10^{-7}$.

\paragraph{Evaluation.}
Because no closed-form noise reference exists for nonlinear reconstructions, convergence was analyzed using method-specific gold-standard references: separate high-sample-count runs at $N = 50{,}000$ for PICO and for PMR, each used to benchmark its own lower-$N$ estimates (Appendix~\ref{app:pmr_convergence}). NRMSE was computed on estimated noise standard deviation maps at each intermediate~$N$; runtime efficiency was quantified by comparing the time for each method to reach converged accuracy. Statistical significance of the NRMSE difference between PICO and PMR was tested using a paired Wilcoxon signed-rank test across the 10 subjects' per-subject mean NRMSE values.

\paragraph{Robustness to input noise level.}
To assess validity of the local linearization, the input noise standard deviation was swept over scale factors $\{1, 5, 10, 50, 100, 200\} \times \sigma_0$ ($\sigma_0 = 10^{-7}$), spanning SNR from approximately $45$ to $-1$\,dB. For each level, variance maps were computed by both methods at $N = 3000$. Normalized output maps ($\boldsymbol{\sigma}_{\hat{\mathbf{x}}}/\sigma_k$) were compared to test whether PICO remains concordant with PMR across a broad SNR range.

\subsection{Implementation details}

\paragraph{Software and operators.}
All methods were implemented in Python/PyTorch with complex64 arithmetic \cite{pytorch}. Non-Cartesian encoding and reconstruction used SigPy \cite{sigpy} (Kaiser--Bessel width 4) with Toeplitz embedding for the normal operator \cite{fessler2005}. Cartesian operators used the same NUFFT framework with grid-aligned trajectories.

\paragraph{Solvers.}
Linear reconstructions used CG on the normal equations; nonlinear CS used FISTA with standard momentum \cite{Beck2009}. Spectral normalization ($\lambda_{\max}$ via 15-step power method, $1.01\times$ margin) was applied to all forward operators.

\paragraph{Hardware and code availability.}
All experiments ran on an NVIDIA RTX 8000 GPU. Reported timings are wall-clock measurements averaged over multiple runs. Code is available at \url{https://github.com/onat-dalmaz/fast_mri_gfactor}.

\section{Results}
\label{sec:results}

\subsection{Cartesian CG-SENSE}

\begin{figure*}[t]
    \centering
    \includegraphics[width=\textwidth]{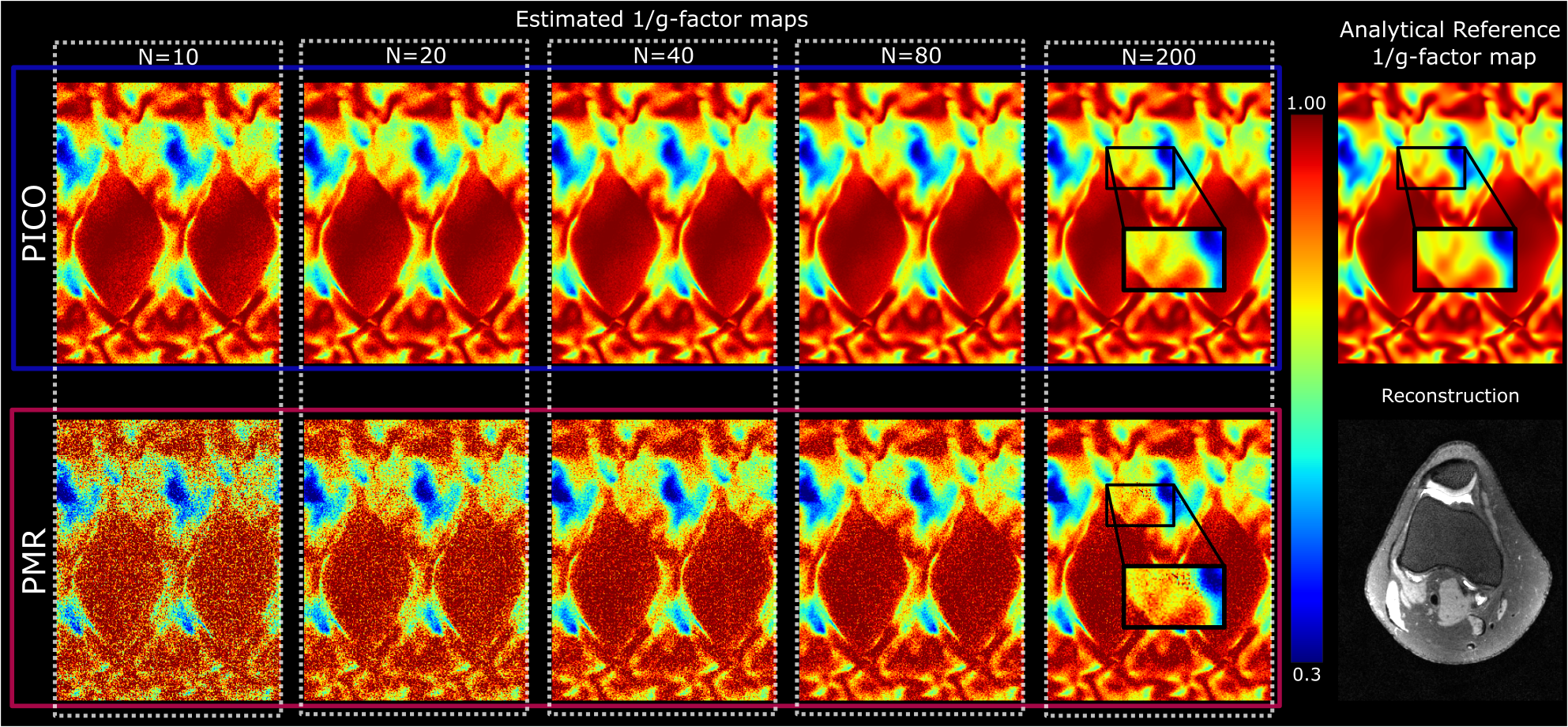}
    \caption{Reciprocal $g$-factor maps ($1/g$) for 
    Cartesian CG-SENSE at $R = 2$, shown for 
    increasing probe/replica counts 
    ($N = 10, 20, 40, 80, 200$). Top: PICO; bottom: 
    PMR. Right column: analytical SENSE reference 
    (top) and reconstructed magnitude image (bottom). 
    All panels share a fixed color scale set by the 
    analytical reference. PICO recovers the spatial 
    structure of the $g$-factor by $N = 10$ and is 
    visually converged by $N = 40$--$80$, whereas 
    PMR retains substantial granular noise 
    through $N = 80$.}
    \label{fig:knee_results}
\end{figure*}

\paragraph{Qualitative results.}
Figure~\ref{fig:knee_results} compares PICO with 
PMR against the closed-form SENSE $g$-factor on 
the Cartesian knee dataset ($R = 2$). PICO 
recovers the global $g$-factor distribution with 
as few as $N = 10$ probes and is visually 
indistinguishable from the analytical reference 
by $N = 40$--$80$; residual error at low $N$ 
manifests as low-frequency bias rather than 
pixel-wise grain. In contrast, PMR exhibits 
spatially uncorrelated noise that obscures fine 
structure even at $N = 80$ and only approaches 
PICO's smoothness at considerably higher replica 
counts.

\paragraph{Quantitative results.}
At $N = 200$, PICO achieved a mean NRMSE of $1.61 \pm 0.19\%$ compared with $8.05 \pm 1.90\%$ for PMR (Wilcoxon $p < 0.01$). Runtime per slice was comparable ($13.4$~s vs.\ $16.0$~s), with the modest advantage arising because each implicit covariance--vector product converges faster than a full CG-SENSE reconstruction.

\subsection{Non-Cartesian CG-SENSE}
\label{subsec:noncart_cgsense_noise_amp}

\begin{figure*}[t]
    \centering
\includegraphics[width=\textwidth]{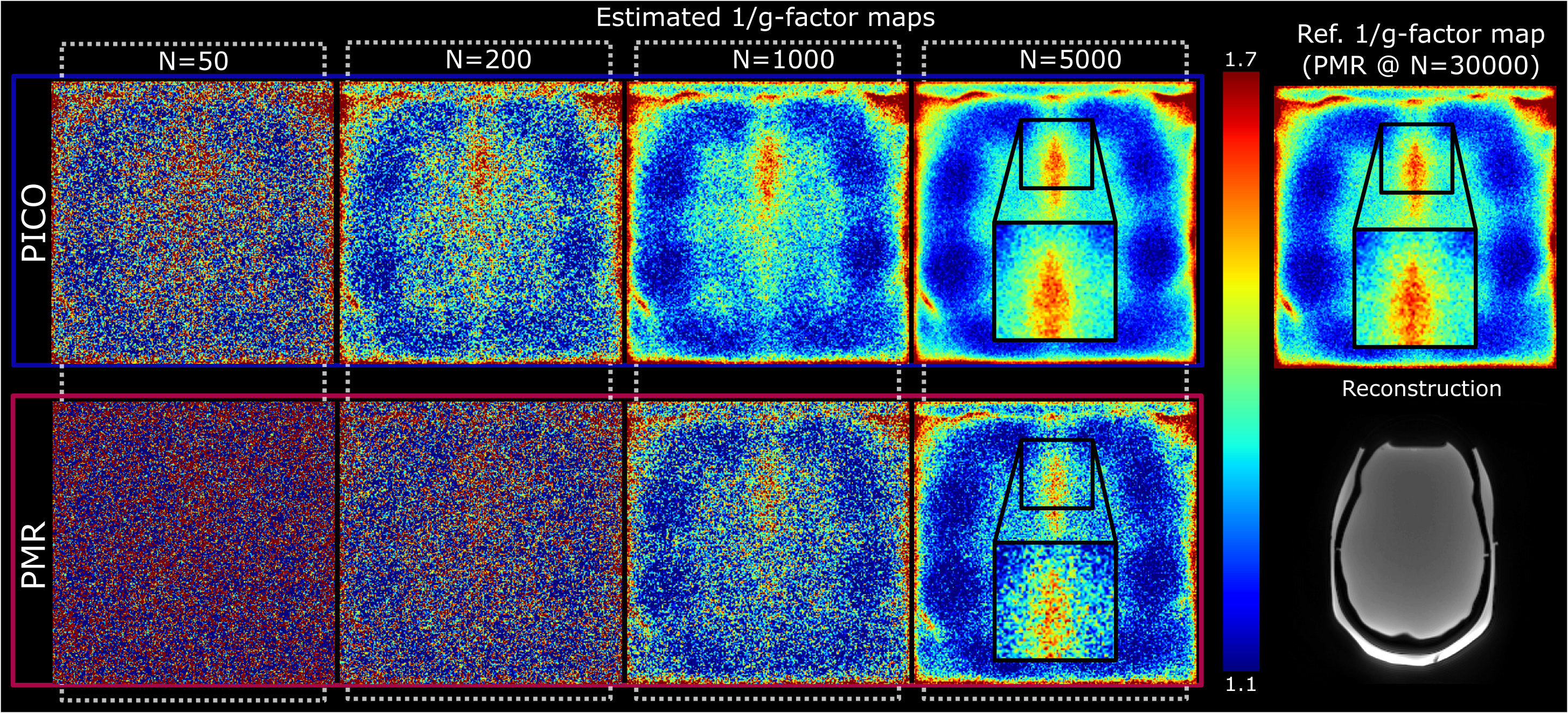}
    \caption{
    Voxelwise $g$-factor maps for non-Cartesian CG-SENSE ($R=2$, spiral trajectory). Columns show increasing sample counts ($N \in \{50, 200, 1000, 5000\}$); right column: PMR reference at $N = 30{,}000$ and magnitude image. Top: PICO; bottom: PMR. PICO resolves the $g$-factor structure by $N = 1000$, whereas PMR retains visible Monte Carlo artifacts even at $N = 5000$.
    }
    \label{fig:spiral_qual}
\end{figure*}

\begin{figure*}[t]
    \centering
    \includegraphics[width=\linewidth]{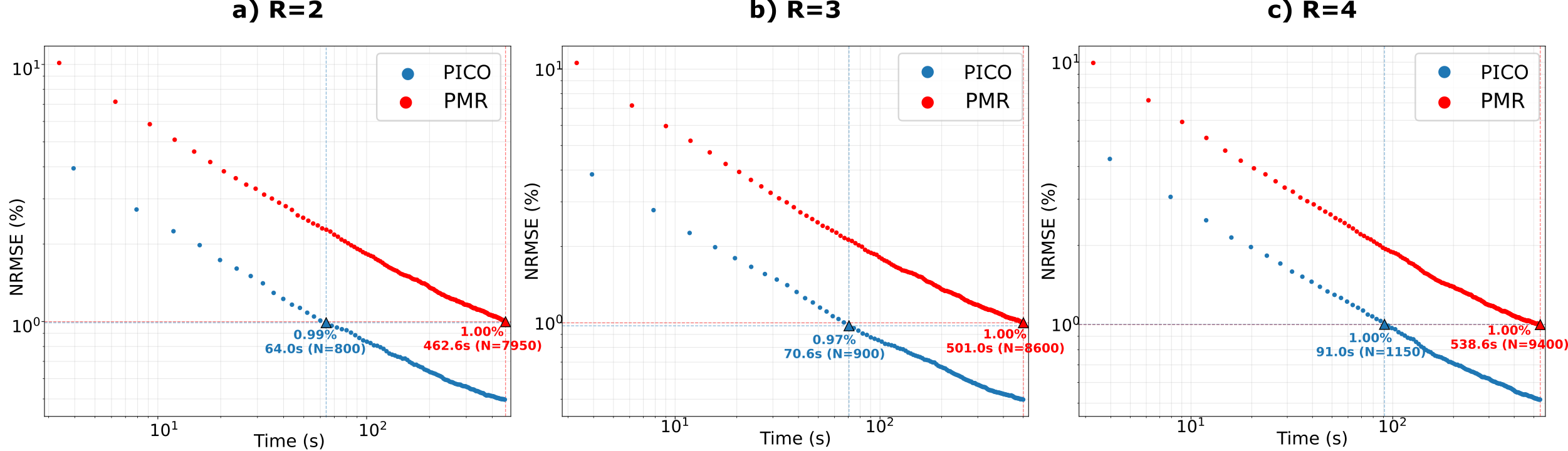}
    \caption{
    NRMSE vs.\ runtime for non-Cartesian CG-SENSE (log--log scale). \textbf{(a)}~$R=2$: PICO reaches $1\%$ NRMSE in $64.0$~s ($N = 800$); PMR requires $462.6$~s ($N = 7950$). \textbf{(b)}~$R = 3$: proposed $70.6$~s ($N = 900$) vs.\ PMR $501.0$~s ($N = 8600$). \textbf{(c)}~$R = 4$: proposed $91.0$~s ($N = 1150$) vs.\ PMR $538.6$~s ($N = 9400$). Dashed line: $1\%$ NRMSE threshold; triangles: first crossing for each method.
    }
    \label{fig:spiral_quant}
\end{figure*}

\paragraph{Qualitative results.}
Figure~\ref{fig:spiral_qual} shows the convergence of $g$-factor estimates for the spiral phantom at $R = 2$. The relatively poor conditioning of the non-Cartesian operator requires higher sample counts than the Cartesian case for both methods. PICO resolves the spatial $g$-factor structure by $N = 1000$; PMR remains dominated by granular noise at the same count and retains visible stochastic texture even at $N = 5000$. Sub-unity $g$-factors ($g < 1$) appear in some voxels; this is an expected consequence of Tikhonov regularization, which can reduce noise variance below the fully sampled reference (Appendix~\ref{app:tik-variance}).

\paragraph{Quantitative accuracy and runtime.}
Figure~\ref{fig:spiral_quant} plots NRMSE against wall-clock time on logarithmic axes for all three acceleration factors. At $R = 2$ (Fig.~\ref{fig:spiral_quant}a), PICO reached $1\%$ NRMSE in $64.0$~s ($N = 800$), compared with $462.6$~s ($N = 7950$) for PMR---an approximately $7.2\times$ speedup. At $R = 3$ (Fig.~\ref{fig:spiral_quant}b), PICO required $N = 900$ probes ($70.6$~s) versus $N = 8600$ replicas ($501.0$~s) for PMR ($\approx\!7.1\times$). At $R = 4$ (Fig.~\ref{fig:spiral_quant}c), the corresponding figures were $N = 1150$ probes ($91.0$~s) versus $N = 9400$ replicas ($538.6$~s; $\approx\!5.9\times$). Across all acceleration factors, going from $R = 2$ to $R = 4$ required only 350 additional probes for PICO versus 1450 additional replicas for PMR, consistent with the theoretical prediction that PMR's estimator variance scales quadratically with local noise variance (Appendix~\ref{app:tail_bounds}).

\paragraph{Probe distribution ablation.}
Figure~\ref{fig:ablation_results} compares the three probe distributions on the spiral dataset. At every $N$, random phase probes yielded the lowest NRMSE, followed by real-valued Rademacher, then complex Gaussian---consistent with the kurtosis-based variance analysis in Appendix~\ref{app:variance-derivation}.

\begin{figure}[t]
    \centering
    \includegraphics[width=0.49\textwidth]{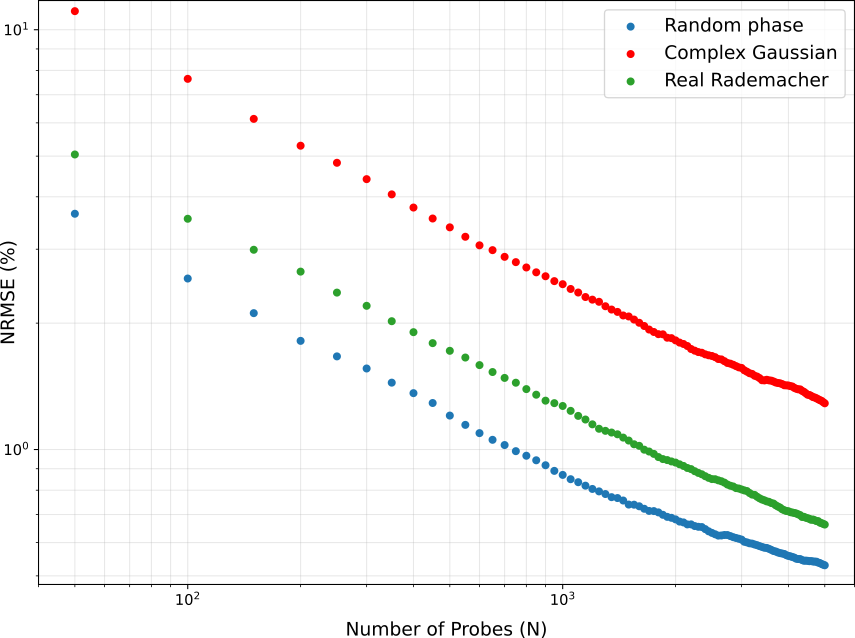}
\caption{
    NRMSE vs.\ probe count $N$ for three probe families on non-Cartesian CG-SENSE (log--log scale). Random phase (blue) achieves the lowest error at every $N$, followed by real Rademacher (green) and complex Gaussian (red), confirming the theoretical optimality of unit-magnitude random-phase probes. All three distributions exhibit approximately parallel power-law decay, with the vertical offset governed by the probe kurtosis (Appendix~\ref{app:variance-derivation}).
    }
    \label{fig:ablation_results}
\end{figure}

\subsection{Compressed-Sensing reconstruction}

\begin{figure*}[t]
    \centering
    \includegraphics[width=\textwidth]{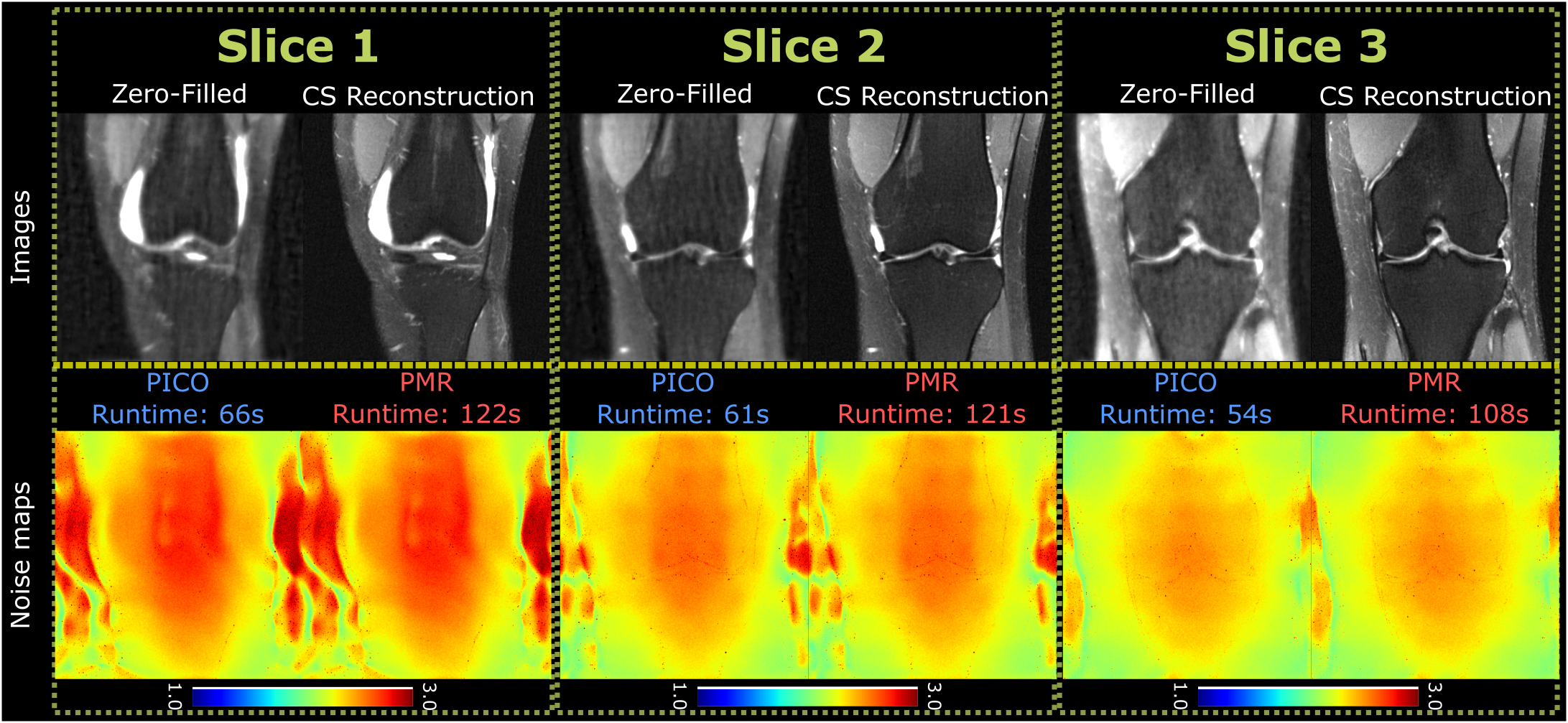}
    \caption{
    Noise standard deviation maps for TV-regularized CS reconstruction ($R=2$). Top: zero-filled and CS magnitude images. Bottom: noise maps from PICO (left) and PMR (right), with slice-specific runtimes. The two methods produce visually identical maps; PICO converges faster (e.g., 66~s vs.\ 122~s for Slice 1).
    }
    \label{fig:cs_efficiency}
\end{figure*}

\paragraph{Accuracy and efficiency.}
Figure~\ref{fig:cs_efficiency} compares noise standard deviation maps from PICO and PMR. The maps are visually indistinguishable, both capturing the data-dependent noise texture characteristic of TV regularization (suppression in uniform regions, preservation near edges). Across the dataset, PICO required $52.3 \pm 7.7$~s compared with $95.5 \pm 14.0$~s for PMR, a statistically significant $\approx\!1.8\times$ speedup (Wilcoxon $p < 0.05$).

\begin{figure*}[t]
    \centering
    \includegraphics[width=\textwidth]{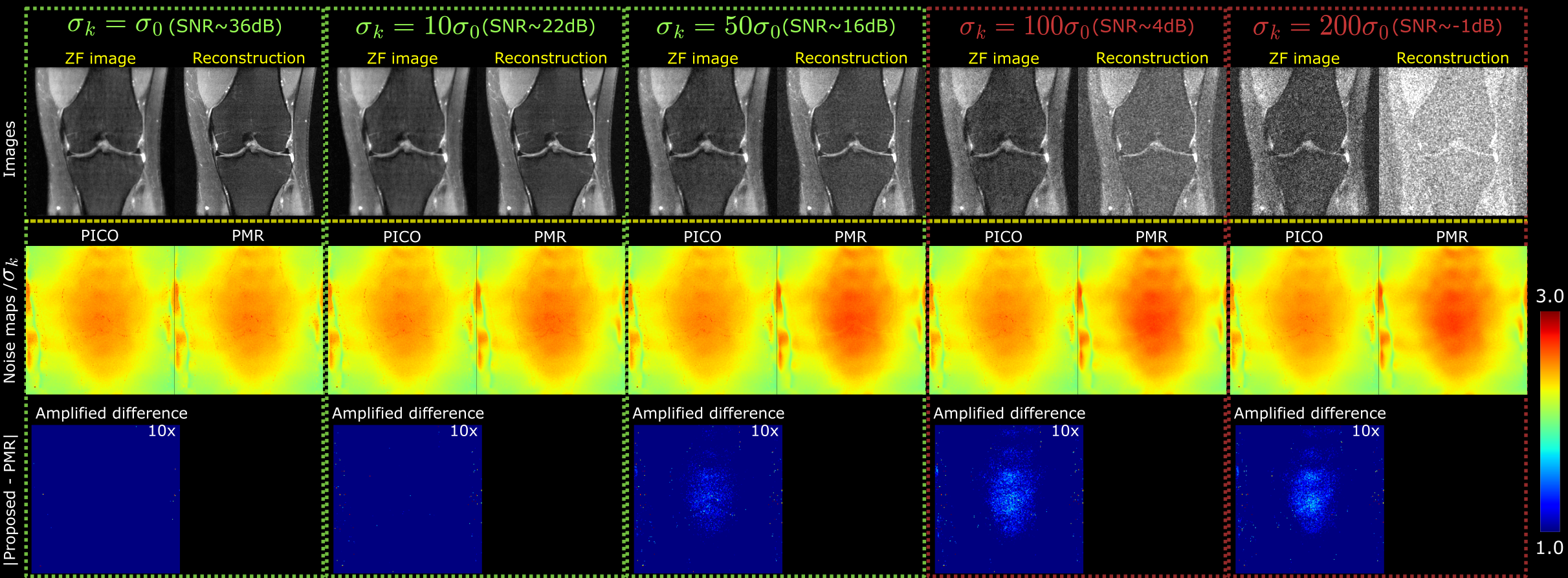}
    \caption{
    Robustness to input noise level in nonlinear CS reconstruction ($N = 3000$). Columns sweep the noise standard deviation from $\sigma_0$ to $200\sigma_0$ (SNR $\approx 45$ to $-1$\,dB). Top: magnitude images. Middle: normalized noise maps ($\sigma_{\mathrm{out}}/\sigma_k$) from PICO and PMR. Bottom: absolute difference and $\times 10$ amplified view. At moderate-to-high SNR ($\ge 16$\,dB), the two methods agree closely. At extreme noise levels ($\le 4$\,dB), residual discrepancies reflect higher-order nonlinear effects beyond the first-order Jacobian model.
    }
    \label{fig:cs_robustness}
\end{figure*}

\paragraph{Robustness to input noise level.}
Figure~\ref{fig:cs_robustness} shows normalized output noise maps across a range of input noise amplitudes (SNR $\approx 45$ to $-1$\,dB). Under a valid local linear model, the normalized map $\boldsymbol{\sigma}_{\hat{\mathbf{x}}}/\sigma_k$ is expected to remain stationary. At moderate-to-high SNR ($\ge 16$\,dB), PICO and PMR exhibit negligible differences, even in $\times 10$ amplified views. At extreme noise levels (SNR $\le 4$\,dB)---well below the SNR range encountered in standard acquisitions \cite{brown2014magnetic,WestbrookMRIinPractice}---the local linear approximation shows visible residual structure, as the active set of the TV proximal operator fluctuates substantially across noise realizations. Even in these stress-test conditions, however, the estimator correctly identifies the dominant heteroscedastic noise amplification patterns, and departures from linearity manifest as second-order effects.
\section{Discussion}

\paragraph{PICO versus PMR: from $k$-space perturbation to image-space probing.}
Image-Space Covariance Probing (PICO) differs from PMR in where randomness enters the pipeline, and this choice drives both the computational and statistical advantages observed in our experiments. 
PMR perturbs 
$k$-space and re-runs a full reconstruction per replica, 
treating the reconstruction as an opaque simulator of noise 
propagation \cite{pmr,dawood,dalmaz2025efficient}. PICO 
shifts the randomness to image space and works with the 
noise-propagation operator directly: in the linear case, 
estimating $\mathrm{diag}(\boldsymbol{\Sigma}_{\hat{\mathbf{x}}})$ 
is reduced to passing unit-magnitude random-phase probe images through 
$\boldsymbol{\Sigma}_{\hat{\mathbf{x}}}$; in the nonlinear 
case, we linearize the reconstruction map at the converged 
solution and use Jacobian--vector products. In both regimes, 
each probe evaluation reuses the same fast primitives already 
built for reconstruction ($\mathbf{A}$, 
$\mathbf{A}^{\mathrm{H}}$, Toeplitz normal operators, inner 
CG steps), avoiding repeated end-to-end solves. A second 
advantage is estimator variance: diagonal 
estimators admit low-variance probe distributions for 
Hermitian positive-semidefinite operators 
\cite{hutchinson,Avron2011}, and unit-magnitude random-phase probes probes 
used by PICO achieve the minimum possible kurtosis ($\kappa = 1$), 
halving the per-sample variance relative to Gaussian probes 
(Appendix~\ref{app:variance-derivation}). Empirically, this 
yields lower estimator variance at fixed $N$ and shorter 
time-to-target accuracy than PMR across all experiments.

\paragraph{Convergence in non-Cartesian settings.}
Both methods converge more slowly in non-Cartesian 
reconstructions than in the Cartesian case, reflecting the 
more globally coupled encoding: NUFFT interpolation and 
trajectory-dependent point-spread sidelobes spread each 
noise perturbation over larger spatial neighborhoods, 
requiring more samples to average out the resulting leakage. 
For PICO, this manifests as increased 
off-diagonal energy in $\boldsymbol{\Sigma}_{\hat{\mathbf{x}}}$, 
which governs the estimator variance 
(Appendix~\ref{app:tail_bounds}). PMR suffers an additional 
penalty: because it estimates variance from squared magnitudes 
of complex-Gaussian outputs, individual replicas in 
high-$g$-factor regions can produce disproportionately large 
fluctuations. These heavy-tailed outliers move the sample 
variance estimate, requiring substantially more replicas 
before the map stabilizes, yielding a substantial efficiency advantage across all tested acceleration factors (Fig.~\ref{fig:spiral_quant}).

\paragraph{Practical implications.}
The speed of PICO enables noise-mapping workflows
that are impractical with PMR. In protocol design, fast noise 
maps enable direct comparison of candidate acquisition 
configurations---e.g., evaluating how acceleration factor, 
phase-encoding direction, or coil layout affects local SNR in 
a target anatomy---without acquiring and reconstructing full 
datasets for each option 
\cite{robson2008snrgfactor,Breuer2009}. This is 
particularly valuable for non-Cartesian trajectories, where 
no closed-form $g$-factor formulas exist and protocol tuning 
has relied on intuition or time-consuming PMR simulations. 
At scan time, noise maps can identify degraded coil elements 
or suboptimal patient positioning before the patient leaves 
the scanner \cite{image_recon_SNR}. For regularized 
reconstructions, voxelwise noise maps offer a principled 
alternative to global heuristics for selecting $\lambda$: 
one can monitor how the spatial noise distribution evolves as 
$\lambda$ varies and select the operating point that achieves 
a desired local SNR floor 
\cite{knoll2017cs_clinical,hanhela2021regparam}. In 
multi-site and longitudinal studies, standardized $g$-factor 
reporting would strengthen reproducibility of quantitative 
biomarkers \cite{boudreau2023ismrm,boscolo2025multisite}; 
PICO could serve as a fast back-end within 
automated SNR evaluation platforms \cite{montin2026mroptimum}, which currently rely on 
PMR-based methods and would benefit from the reduced 
computation time.
More broadly, PICO could become a 
computational tool for the wide range of SNR 
characterization studies that have so far relied on 
simple Cartesian acquisitions with closed-form noise 
models---for example, the relative-SNR comparisons between 
supine and prone breast coil configurations \cite{hess2026breastsnr}, which currently 
depend on Cartesian body-coil normalization and would benefit 
from voxelwise noise maps that generalize to arbitrary 
trajectories and reconstruction methods.

Because PICO averages independent probe contributions, the empirical variance across probes provides a natural online convergence diagnostic: one can monitor the per-voxel standard error of the running mean and stop adding probes when it falls below a user-defined threshold. In the present work, we used fixed~$N$ for clean benchmarking against PMR, but an adaptive stopping rule based on the estimator's own variance would reduce computational cost in practice and could be implemented without additional reconstruction calls. The theoretical variance decomposition (Appendix~\ref{app:variance-derivation}, Eq.~\eqref{eq:var_single_sample}) quantifies the expected convergence rate as a function of probe count and covariance structure, providing a principled basis for such adaptive strategies.

Additional computational savings would be possible by estimating
variance on a reduced-resolution grid---noise maps are
governed primarily by coil geometry and vary more smoothly
than anatomical content---and interpolating back to native
resolution, a strategy compatible with both the proposed
estimator and PMR but most effective when combined with the
lower per-probe cost of PICO.
%

\paragraph{Extensions beyond voxelwise variance.}
This work focuses on the diagonal of the image-domain 
covariance, consistent with the standard definition of the 
$g$-factor and most prior noise-analysis literature 
\cite{sense,pmr,Breuer2009,akcakaya2014,dalmaz2025efficient}. 
However, the off-diagonal entries---capturing spatial 
cross-correlations between voxels---encode complementary 
information relevant to applications such as functional 
connectivity estimation in fMRI, where structured covariance 
is the central object of interest 
\cite{varoquaux2010covsel,behrouz2024bridgefc}. Extending the 
framework to approximate selected off-diagonal terms or 
low-rank covariance structure could link acquisition and 
reconstruction choices to downstream connectivity metrics. 
A related direction is propagating reconstruction-induced 
noise through parameter-estimation pipelines (e.g., $T_1/T_2$ 
mapping, diffusion) to obtain spatiotemporal uncertainty maps 
rather than a single static $g$-factor slice; recent work on 
dictionary-based quantitative MRI has demonstrated the value 
of explicitly quantifying matching uncertainty 
\cite{toner2026uqdict}.

\paragraph{Leveraging noise maps in downstream imaging tasks.}
Beyond reconstruction diagnostics, voxelwise noise-variance 
maps can be consumed directly by downstream processing. 
Classical adaptive filters (e.g., BM3D) use local variance 
to tune denoising thresholds \cite{hanchate2020,lidenoising2014}, 
and spatially varying noise estimates have improved 
content in low-dose CT and ultrasound 
\cite{Hariharan_2020,yu2002}. In deep-learning settings, 
feeding noise or uncertainty maps as additional input channels 
or using them for loss reweighting has been reported to reduce 
false positives and improve calibration 
\cite{dou2025,wang2022,homolya2025raki}. Recent work has demonstrated that providing 
$g$-factor maps directly as network inputs---together with 
SNR-unit-normalized images---yields architecture-agnostic 
denoising improvements and strong generalization across field 
strengths, contrasts, and anatomies 
\cite{xue2025snraware,xue2025imtransformer}. Those methods 
rely on closed-form GRAPPA-based $g$-factor computation, 
which is confined to Cartesian acquisitions, and revert to PMR 
for measuring output noise when evaluating the trained 
denoisers. PICO extends this capability in two ways: first, by making voxelwise noise maps available for non-Cartesian and nonlinear reconstruction pipelines where analytical $g$-factor formulas are unavailable, potentially enabling SNR-aware denoising for a broader class of acquisitions; and second, by replacing the PMR-based output-noise measurement step with a faster and lower-variance alternative.
Systematically coupling fast noise estimation with 
variance-aware filters and networks, and quantifying the 
benefit over simpler heuristic noise models under aggressive 
acceleration, is an important direction for future work.

\paragraph{Limitations.}
Although PICO achieves improved $g$-factor 
and noise estimation relative to PMR under matched computation 
time, this work does not directly assess how the improvements 
translate to observer perception or task-based diagnostic 
performance. Classical parallel-imaging studies have evaluated 
$g$-factor maps as physics-based summary metrics rather than 
tools inspected alongside images 
\cite{pruessmann1999sense,robson2008snrgfactor,Breuer2009}, 
yet broader literature on task-based image quality shows that 
objective noise measures can diverge from radiologist 
preferences \cite{barrett2015task,kastryulin2023iqamri,
ma2024radiqmri,tang2025rankiqmri}. Prospective reader studies 
that define application-specific operating points---how much 
estimator accuracy is needed before perceived noise texture 
and lesion conspicuity stop improving---would help align noise 
estimation research with downstream decision-making. A second 
limitation is the Jacobian-based approximation for nonlinear 
reconstructions, which is valid to first order in the noise 
amplitude. Our robustness experiments 
(Fig.~\ref{fig:cs_robustness}) confirm close agreement with 
PMR at moderate-to-high SNR, but at extreme noise levels the 
local linear model shows visible residual error. Extending 
the framework to capture higher-order nonlinear effects, or 
to provide formal uncertainty bounds on the linearization 
error itself, remains an open problem. 
A third direction is operator-aware probe design: 
random-phase probes are optimal among distributions 
(Appendix~\ref{app:variance-derivation}), but 
optimized probes that exploit the covariance's structure~\cite{bekas} 
could further reduce estimator variance.

\section{Conclusion}
We have presented a stochastic diagonal estimation framework 
for fast, voxelwise noise variance and $g$-factor mapping in 
MRI reconstructions. For linear reconstructions such as 
CG-SENSE---including the unaccelerated case---the method 
recovers the exact diagonal of the image-domain noise 
covariance without any model approximation; for nonlinear 
reconstructions (e.g., compressed sensing with TV 
regularization), it extends naturally through Jacobian--vector 
products. In both settings, unit-magnitude random-phase probes provably minimize the per-sample estimator variance.
 
Across Cartesian, non-Cartesian, and compressed-sensing 
experiments, PICO matched or exceeded the 
accuracy of high-replica PMR references while requiring a 
fraction of the computation time. These results demonstrate 
that rigorous noise characterization need not be a 
computational bottleneck: the same operator primitives used 
for image reconstruction can be reused to produce voxelwise 
SNR maps as a routine by-product of the reconstruction 
pipeline. We anticipate that embedding fast noise estimation 
into standard MRI workflows will support protocol 
optimization, quality assurance, and downstream 
uncertainty-aware analysis across a broad range of 
reconstruction methods.

\section*{Data Availability Statement}
All image data, source code for reconstruction experiments, and noise estimation toolbox are publicly available at \url{https://https://github.com/onat-dalmaz/pico_noise_estimation_toolbox}.

\bibliography{MRM-AMA}%
\vfill\pagebreak
 
\section*{Supporting Information}
 
\appendix

\section{PICO Algorithm Pseudocode}
\label{app:PICO_algorithm}

Algorithm~\ref{alg:PICO} summarizes the PICO procedure 
as implemented in all experiments reported in the main 
text. The algorithm is stated for the linear case, 
where the covariance operator 
$\boldsymbol{\Sigma}_{\hat{\mathbf{x}}} 
= \mathbf{R}\mathbf{R}^{\mathrm{H}}$ is applied 
implicitly via forward/adjoint primitives and a CG solve 
on the normal equations. For nonlinear reconstructions, 
the covariance--vector product 
$\boldsymbol{\Sigma}_{\hat{\mathbf{x}}}\mathbf{v}^{(i)}$ 
in line~\ref{line:cov_apply} is replaced by a 
Jacobian--vector product 
$\mathbf{J}_f(\mathbf{k}_0)\mathbf{v}^{(i)}$ evaluated 
via automatic differentiation at the converged 
reconstruction (Sec.~\ref{subsec:cs_extension}); the 
remainder of the algorithm is unchanged.

\begin{algorithm}[t]
\caption{Probing Image-space COvariance (PICO)}
\label{alg:PICO}
\begin{algorithmic}[1]
\Require Encoding operator $\mathbf{A}$ (and adjoint 
         $\mathbf{A}^{\mathrm{H}}$); regularization 
         parameter $\lambda \ge 0$; number of probes $N$; 
         image dimension $N_x$.
\Ensure  Estimated voxelwise noise variance map 
         $\hat{\boldsymbol{\sigma}}^{2}_{\hat{\mathbf{x}}} 
         \approx 
         \mathrm{diag}(\boldsymbol{\Sigma}_{\hat{\mathbf{x}}})$.
\State $\mathbf{s} \gets \mathbf{0} \in \mathbb{C}^{N_x}$ 
       \Comment{Running sum across probes}
\For{$i = 1$ \textbf{to} $N$}
    \State \Comment{(1) Draw a unit-magnitude random-phase probe}
    \State $\theta_k \sim \mathrm{Uniform}[0, 2\pi]$ 
           for $k = 1, \dots, N_x$
    \State $v^{(i)}_k \gets e^{j\theta_k}$ 
           for $k = 1, \dots, N_x$
    \State \Comment{(2, 3) Apply $\boldsymbol{\Sigma}_{\hat{\mathbf{x}}}$ implicitly via CG}
    \State $\mathbf{u}^{(i)} \gets 
           \boldsymbol{\Sigma}_{\hat{\mathbf{x}}}\,\mathbf{v}^{(i)}$ 
           via CG on the normal equations 
           \label{line:cov_apply}
    \State \Comment{(4) Accumulate elementwise product}
    \State $\mathbf{s} \gets \mathbf{s} 
           + \mathbf{v}^{(i)*} \odot \mathbf{u}^{(i)}$
\EndFor
\State $\hat{\boldsymbol{\sigma}}^{2}_{\hat{\mathbf{x}}} 
       \gets \mathrm{Re}\!\left(\mathbf{s} / N\right)$ 
       \Comment{Imaginary part vanishes in expectation}
\end{algorithmic}
\end{algorithm}

Each probe response in line~\ref{line:cov_apply} is 
obtained by applying $\boldsymbol{\Sigma}_{\hat{\mathbf{x}}} 
= \mathbf{R}\mathbf{R}^{\mathrm{H}}$ implicitly through 
the same forward/adjoint primitives 
($\mathbf{A}$, $\mathbf{A}^{\mathrm{H}}$, FFT/NUFFT, 
Toeplitz embedding, coil sensitivity maps) and CG 
solver already present in the CG-SENSE reconstruction 
pipeline. The reconstruction operator $\mathbf{R}$ and 
the covariance matrix 
$\boldsymbol{\Sigma}_{\hat{\mathbf{x}}}$ are never 
formed explicitly. Computational cost is dominated by 
the per-probe CG work in line~\ref{line:cov_apply}, 
which reuses the same inner solver as CG-SENSE 
reconstruction.
 
\section{Variance of the Stochastic Diagonal Estimator}
\label{app:variance-derivation}
 
Let 
$\boldsymbol{\Sigma}_{\hat{\mathbf{x}}} 
= \mathbf{R}\mathbf{R}^{\mathrm{H}} \in \mathbb{C}^{N_x \times N_x}$
be positive-semidefinite (and necessarily Hermitian by construction), and define for a single probe image
$\mathbf{v} \in \mathbb{C}^{N_x}$ with i.i.d.\ entries satisfying
$\mathbb{E}[v_k] = 0$,
$\mathbb{E}[|v_k|^2] = 1$, and
$\mathbb{E}[v_k^{*}v_\ell] = 0$ for $k \neq \ell$,
the voxelwise point estimator
\[
  \delta_k(\mathbf{v})
  = v_k^{*}\,
  \bigl(\boldsymbol{\Sigma}_{\hat{\mathbf{x}}}\,\mathbf{v}\bigr)_k,
  \qquad k = 1,\dots,N_x.
\]
Because $\mathbb{E}[v_k^{*}v_\ell] = \delta_{k\ell}$,
the estimator is unbiased:
$\mathbb{E}[\delta_k(\mathbf{v})]
= (\boldsymbol{\Sigma}_{\hat{\mathbf{x}}})_{kk}$.
 
\paragraph{Single-sample variance.}
Let $\kappa := \mathbb{E}[|v_k|^{4}]$ denote the fourth moment of
the probe distribution. Expanding
$(\boldsymbol{\Sigma}_{\hat{\mathbf{x}}}\mathbf{v})_k
= (\boldsymbol{\Sigma}_{\hat{\mathbf{x}}})_{kk}v_k
+ \sum_{\ell \neq k} (\boldsymbol{\Sigma}_{\hat{\mathbf{x}}})_{k\ell} v_\ell$
and using independence yields
\begin{equation}
\operatorname{Var}\!\bigl[\delta_k(\mathbf{v})\bigr]
  = (\kappa - 1)\,
    \bigl|(\boldsymbol{\Sigma}_{\hat{\mathbf{x}}})_{kk}\bigr|^{2}
   + \sum_{\ell \neq k}
    \bigl|(\boldsymbol{\Sigma}_{\hat{\mathbf{x}}})_{k\ell}\bigr|^{2}.
\tag{A1}\label{eq:var_single_sample}
\end{equation}

The per-sample variance separates into two contributions: a \emph{diagonal-dependent term} $(\kappa - 1)\,|(\boldsymbol{\Sigma}_{\hat{\mathbf{x}}})_{kk}|^2$ that is controlled by the kurtosis $\kappa$ of the probe distribution and scales with the squared local noise variance, and an \emph{off-diagonal interference term} $\sum_{\ell \neq k}|(\boldsymbol{\Sigma}_{\hat{\mathbf{x}}})_{k\ell}|^2$ that is invariant across probe distributions and depends only on the noise leakage from neighboring voxels into voxel $k$. Minimizing the per-sample variance subject to the unit-variance constraint $\mathbb{E}[|v_k|^2] = 1$ therefore reduces to minimizing $\kappa$, which determines the choice of probe distribution analyzed below.

\paragraph{Choice of probe distribution.}
Three distributions satisfying the probe conditions are compared:
\begin{itemize}
\item \emph{Complex Gaussian probes}
  ($v_k \sim \mathcal{CN}(0,1)$)
  yield $\kappa = 2$, contributing one unit of the diagonal squared magnitude to the per-sample variance for every probe.
\item \emph{Real $\pm 1$ Rademacher probes}
  ($v_k \in \{-1,+1\}$ with $\mathbb{P}(v_k = +1) = \mathbb{P}(v_k = -1) = 1/2$) attain $\kappa = 1$, eliminating the diagonal-dependent contribution. Since $\kappa \ge 1$ for any unit-variance distribution (by Jensen's inequality), this is the kurtosis-optimal real-valued choice.
\item \emph{Unit-magnitude random-phase probes}
  ($v_k = e^{\mathrm{i}\theta_k}$,
  $\theta_k \sim \mathsf{U}[0,2\pi]$, so that $|v_k| = 1$ deterministically while only the phase varies across draws) also attain $\kappa = 1$, matching the kurtosis of real Rademacher probes. Unlike real Rademacher, however, each unit-magnitude random-phase sample carries two real degrees of freedom---the sine and cosine of a uniformly distributed phase---which permits the estimator to interact with the imaginary components of the off-diagonal entries $(\boldsymbol{\Sigma}_{\hat{\mathbf{x}}})_{k\ell}$. For the complex Hermitian covariances arising in MRI reconstruction, this doubles the information extracted per sample relative to real Rademacher and halves the number of probes required to reach a given target accuracy.
\end{itemize}

\noindent\textbf{Proposition} (Optimal probe distribution).
Under the constraints
$\mathbb{E}[\mathbf{v}] = \mathbf{0}$ and
$\mathbb{E}[\mathbf{v}\mathbf{v}^{\mathrm{H}}] = \mathbf{I}$,
unit-magnitude random-phase probes minimize the per-sample
variance $\operatorname{Var}[\delta_k(\mathbf{v})]$ in \eqref{eq:var_single_sample} for complex Hermitian positive-semidefinite $\boldsymbol{\Sigma}_{\hat{\mathbf{x}}}$, and are therefore
optimal for the proposed stochastic diagonal estimator.

\section{Theoretical Error Bounds and Convergence Rates}
\label{app:tail_bounds}
 
We analyze the concentration properties of the two estimators 
using standard inequalities for sums of independent random 
variables \cite{Vershynin_2026}.
 
\subsection{PMR: sub-exponential concentration}
The PMR estimator $\hat{\sigma}^2_{\mathrm{PMR}}$ is the 
sample variance of $N$ complex-Gaussian reconstruction 
outputs. The squared magnitude of each output follows an 
exponential distribution, so the sum follows a Gamma 
distribution with sub-exponential tails. Applying the Bernstein inequality for sums of independent sub-exponential random variables \cite[Theorem~2.8.1]{Vershynin_2026} yields
\begin{equation}
\mathbb{P}\!\left( 
  \bigl| \hat{\sigma}^2_{\mathrm{PMR}} - \sigma_k^2 \bigr| 
  \ge t \right) 
\le 2 \exp\!\left( -c\,N \min\!\left( 
  \frac{t^2}{(\sigma_k^2)^2},\; 
  \frac{t}{\sigma_k^2} \right) \right).
\end{equation}
The decay rate is controlled by $(\sigma_k^2)^2$: as 
acceleration increases, local noise variance $\sigma_k^2$ 
grows via the $g$-factor, and the probability of large 
deviations increases quadratically. PMR is therefore 
disproportionately unstable in high-$g$-factor regions.
 
\subsection{PICO: sub-Gaussian concentration}
PICO uses unit-magnitude random-phase probes with
$|v_k| = 1$ deterministically. The diagonal-dependent contribution to the single-sample variance (Eq.~\eqref{eq:var_single_sample}) vanishes because the kurtosis term $(\kappa - 1)$ is zero; estimation error
arises solely from off-diagonal interference. Since the real and imaginary parts of each probe component are bounded on $[-1,1]$, the per-voxel estimator $\delta_k(\mathbf{v})$ is a sum of bounded random variables, and Hoeffding's inequality \cite[Theorem~2.2.5]{Vershynin_2026} applied to the real and imaginary parts gives
\begin{equation}
\mathbb{P}\!\left( 
  \bigl| \hat{\sigma}^2_{\mathrm{Prop}} - \sigma_k^2 \bigr| 
  \ge t \right) 
\le 2 \exp\!\left( 
  -\frac{2\,N\,t^2}
  {\sum_{j \neq k} 
   |(\boldsymbol{\Sigma}_{\hat{\mathbf{x}}})_{kj}|^2} 
\right).
\end{equation}
The bound depends on the off-diagonal energy (noise leakage) 
but is independent of the diagonal variance $\sigma_k^2$ 
itself. At low acceleration, the covariance is nearly 
diagonal and the denominator is small, yielding rapid 
convergence. At high acceleration, off-diagonal energy 
increases, but the estimator does not incur the 
$(\sigma_k^2)^2$ penalty of PMR, and the sub-Gaussian decay 
($\exp(-t^2)$) ensures tighter concentration than the 
sub-exponential tails characteristic of PMR. This confirms a 
fundamental stability advantage: PICO 
decouples estimation precision from the local noise level.

\section{Variance Shrinkage under Tikhonov Regularization}
\label{app:tik-variance}
 
We work in the pre-whitened setting of 
Sec.~\ref{sec:theory}, where
$\mathbf{b}_w = \mathbf{A}\mathbf{x} + \mathbf{n}_w$
with $\mathbf{n}_w \sim \mathcal{CN}(\mathbf{0}, \mathbf{I})$.
Let $\mathbf{M} := \mathbf{A}^{\mathrm{H}}\mathbf{A} 
\succeq \mathbf{0}$.
 
\begin{proposition}[Modewise variance shrinkage]
\label{prop:tik-variance-shrink-app}
Assume $\mathbf{M}$ is positive-definite, and fix $\lambda > 0$. The unregularized least-squares and Tikhonov reconstruction operators
\[
\mathbf{R}_{\mathrm{LS}}
= (\mathbf{A}^{\mathrm{H}}\mathbf{A})^{-1}
  \mathbf{A}^{\mathrm{H}},
\qquad
\mathbf{R}_{\lambda}
= (\mathbf{A}^{\mathrm{H}}\mathbf{A}+\lambda\mathbf{I})^{-1}
  \mathbf{A}^{\mathrm{H}}
\]
have image-domain noise covariances
\[
\boldsymbol{\Sigma}_{\hat{\mathbf{x}},\mathrm{LS}}
= \mathbf{M}^{-1},
\qquad
\boldsymbol{\Sigma}_{\hat{\mathbf{x}},\lambda}
= (\mathbf{M}+\lambda\mathbf{I})^{-1}
  \mathbf{M}\,
  (\mathbf{M}+\lambda\mathbf{I})^{-1}.
\]
Then 
$\boldsymbol{\Sigma}_{\hat{\mathbf{x}},\lambda} 
\prec \boldsymbol{\Sigma}_{\hat{\mathbf{x}},\mathrm{LS}}$
in the Loewner order for every $\lambda > 0$, and 
consequently
$\mathrm{diag}(\boldsymbol{\Sigma}_{\hat{\mathbf{x}},\lambda})
< \mathrm{diag}(\boldsymbol{\Sigma}_{\hat{\mathbf{x}},\mathrm{LS}})$
elementwise.
\end{proposition}
 
\begin{proof}
Let $\mathbf{M} 
= \mathbf{U}\boldsymbol{\Lambda}\mathbf{U}^{\mathrm{H}}$
with eigenvalues $\lambda_i > 0$. Define
$\mathbf{X} 
:= \mathbf{M}^{1/2}(\mathbf{M}+\lambda\mathbf{I})^{-1}
   \mathbf{M}^{1/2}
= \mathbf{U}\,
  \mathrm{diag}\!\bigl(\tfrac{\lambda_i}{\lambda_i+\lambda}\bigr)
  \mathbf{U}^{\mathrm{H}}$.
Then
$\boldsymbol{\Sigma}_{\hat{\mathbf{x}},\lambda}
= \mathbf{M}^{-1/2}\mathbf{X}^2\mathbf{M}^{-1/2}$
and
$\boldsymbol{\Sigma}_{\hat{\mathbf{x}},\mathrm{LS}}
= \mathbf{M}^{-1/2}\mathbf{I}\,\mathbf{M}^{-1/2}$.
Since $0 < \tfrac{\lambda_i}{\lambda_i+\lambda} < 1$ for all 
$\lambda_i > 0$, we have 
$\mathbf{0} \prec \mathbf{X}^2 \prec \mathbf{I}$.
Congruence with $\mathbf{M}^{-1/2}$ preserves the order, 
giving the result.
\end{proof}
 
\begin{corollary}[Regularization can yield $g < 1$]
\label{cor:g-less-than-one-app}
With $g(i) 
= \sigma_{\hat{\mathbf{x}},\mathrm{acc},\lambda}(i) 
  \big/ \bigl(\sqrt{R}\,
  \sigma_{\hat{\mathbf{x}},\mathrm{ref}}(i)\bigr)$,
classical unregularized SENSE implies $g \ge 1$.
Tikhonov regularization introduces shrinkage factors 
$(\lambda_j / (\lambda_j + \lambda))^2 < 1$ along each 
singular mode $j$. Where this shrinkage outweighs the 
geometric noise inflation from undersampling, 
$g(i) < 1$ is observed.
\end{corollary}
 
\paragraph{Remarks.}
(i) Under non-white acquisition noise, all statements hold 
after pre-whitening with 
$\mathbf{A} = \mathbf{L}^{-1}\widetilde{\mathbf{A}}$ and 
$\mathbf{n}_w = \mathbf{L}^{-1}\mathbf{n}$, where 
$\mathbf{L}\mathbf{L}^{\mathrm{H}} 
= \boldsymbol{\Sigma}_{\mathbf{n}}$.
(ii) If $\mathbf{M}$ is rank-deficient, the inequalities are 
strict on its positive spectrum; nullspace components do not 
affect voxel variances within the range of 
$\mathbf{A}^{\mathrm{H}}$.

\section{PMR Convergence Protocol and Reference Budgets}
\label{app:pmr_convergence}
 
For settings without a closed-form noise model, surrogate 
reference variance maps were defined by high-replica PMR 
\cite{pmr}. To certify convergence of the reference, the 
following protocol was applied: (1)~compute a gold-standard 
PMR map at $N_{\max}$ replicas; (2)~evaluate voxelwise NRMSE 
and the change in ROI-mean variance ($\Delta$ROI) at each 
intermediate~$N$ relative to the gold map; (3)~declare 
convergence at the smallest~$N$ such that doubling the 
replica count improves NRMSE by ${<}\,0.5\%$ and $\Delta$ROI 
by ${<}\,0.2\%$.
 
\paragraph{Non-Cartesian spiral physical phantom.}
For the prospectively acquired non-Cartesian spiral brain phantom (a physical phantom imaged on a GE 3T scanner; see Sec.~\ref{sec:methods}), the gold map was 
computed at $N_{\max} = 50{,}000$ replicas ($4{,}340$~s). 
The reference used in the main-text experiments was set at 
$N_{\mathrm{ref}} = 30{,}000$ replicas ($2{,}604$~s), at 
which point the NRMSE relative to the gold was 
$\approx 0.4\%$---well below $1\%$---and 
$\Delta$ROI~$\approx 0.001\%$. Increasing the replica count 
from $N_{\mathrm{ref}}$ to $N_{\max}$ produced negligible 
further improvement, confirming that the 
$N_{\mathrm{ref}} = 30{,}000$ reference is adequately 
converged for the $1\%$ NRMSE operating points reported in 
the main text.
 
\paragraph{Nonlinear TV-regularized compressed sensing.}
Because PICO's Jacobian-based approximation and PMR converge to slightly different asymptotic limits in the nonlinear regime (the former to the first-order Jacobian covariance, the latter to the true nonlinear variance), a single shared reference is not appropriate. Method-specific gold maps were instead computed at 
$N_{\max} = 50{,}000$ ($1{,}614$~s) for both the proposed 
Jacobian-based estimator and PMR. Convergence was reached 
at $N_{\mathrm{ref}} = 5{,}000$ 
(NRMSE~$\approx 1.5\%$, 
$\Delta$ROI~$\approx 0.1\%$; wall-clock $161.4$~s). 
Doubling to $N = 10{,}000$ reduced NRMSE to 
$\approx 1.2\%$, an improvement of $<0.5\%$, confirming 
adequate convergence at $N_{\mathrm{ref}}$.
 
These converged reference maps define the baselines against 
which NRMSE is evaluated in the main text. Because each 
method is assessed against its own converged limit, the 
reported NRMSE curves reflect internal convergence rates 
rather than absolute accuracy differences between estimators.

\end{document}